\documentclass{llncs}

\usepackage{amsmath}
\usepackage{amssymb}

\usepackage{xspace}

\usepackage{hyperref}

\usepackage{paralist}
\usepackage{graphicx}
\graphicspath{{.}{graphics/}}
\usepackage[margin=1cm,font=small]{caption}

\usepackage{flafter}

\usepackage{tikz}
\newtheorem{observation}[theorem]{Observation}
\newtheorem{hypothesis}[theorem]{Hypothesis}

\newcommand{\mc}[1]{\mathcal{#1}}

\newcommand{\hide}[1]{}
\newcommand{\sgraph}[1]{\ensuremath{\mathcal{SG}
\ifthenelse{\equal{#1}{}}{}{(#1)}
}}
\newcommand{\cgraph}[1]{\ensuremath{\mathcal{CG}
\ifthenelse{\equal{#1}{}}{}{(#1)}
}}
\newcommand{\cpath}[2]{\ensuremath{P_{#1}
\ifthenelse{\equal{#2}{}}{}{(#2)}
}}

\newcommand{\problemmacro}[1]{\texorpdfstring{\textsc{#1}}{#1}\xspace}
\newcommand{\Pprecpmtn}{\textsc{P}$|$\textsf{prec,\xspace pmtn,}\xspace$p_j=1\xspace|$\textsc{C}${}_{max}$\xspace}
\newcommand{\Pprecpmtnnn}{\textsc{P}$|$\textsf{prec,\xspace pmtn}\xspace$\xspace|$\textsc{C}${}_{max}$\xspace}
\newcommand{\Pprec}{\textsc{P}$|$\textsf{prec}\xspace$|$\textsc{C}${}_{max}$\xspace}
\newcommand{\Qprec}{\textsc{Q}$|$\textsf{prec}\xspace$|$\textsc{C}${}_{max}$\xspace}
\newcommand{\Oprec}{\textsc{1}$|$\textsf{prec}\xspace$|\sum_{j}w_j$\textsc{C}${}_{j}$\xspace}
\newcommand{\uniquegames}{\problemmacro{unique games}}
\newcommand{\dtooneg}{{\textsf{d-to-1} \textsc{Games}}\xspace}
\newcommand{\dtoone}{{\textsf{d-to-1}}\xspace}

\begin{document}
\title{Towards Tight Lower Bounds for Scheduling Problems}
%Approximating (Dynamic) Facility Location \\ via Exponential Clocks
%Dynamic Facility Location via Exponential Clocks
%\newline\Hnote{New LP-rounding algorithms for the classic and dynamic facility location problems}
%}}
\author{Abbas Bazzi\inst{1} \and Ashkan Norouzi-Fard\inst{2}}
\institute{School of Computer and Communication Sciences, EPFL. \email{abbas.bazzi@epfl.ch}  \and School of Computer and Communication Sciences, EPFL. \email{ashkan.norouzifard@epfl.ch}  }
%\email{abbas.bazzi@epfl.ch \and ashkan.norouzifard@epfl.ch}
%\author{Abbas Bazzi\thanks{School of Computer and Communication Sciences, EPFL.
%Email:
%\href{mailto:abbas.bazzi@epfl.ch}{abbas.bazzi@epfl.ch}.},
%Ashkan Norouzi-Fard\thanks{School of Computer and Communication Sciences, EPFL.
%Email:
%\href{mailto:ashkan.norouzifard@epfl.ch}{ashkan.norouzifard@epfl.ch}.}}
%Ola Svensson\thanks{School of Computer and Communication Sciences, EPFL.
%Email:
%\href{mailto:ola.svensson@epfl.ch}{ola.svensson@epfl.ch}. Supported by ERC Starting Grant 335288-OptApprox.}}

%\setcounter{section}{-1}

%\includeonly{connFlowsCuts}
%\includeonly{colorings}
%\includeonly{planarity}

\setcounter{page}{0}
\maketitle
\thispagestyle{empty}
\begin{abstract}
We show a close connection between structural hardness for $k$-partite graphs and tight inapproximability results for scheduling problems with precedence constraints. Assuming a natural but nontrivial generalisation of the bipartite structural hardness result of~\cite{BK09}, we obtain a hardness of $2-\epsilon$ for the problem of minimising the makespan for scheduling precedence-constrained jobs with preemption on identical parallel machines. This matches the best approximation guarantee for this problem~\cite{Gra66,GR08}. Assuming the same hypothesis, we also obtain a super constant inapproximability result for the problem of scheduling precedence-constrained jobs on related parallel machines, making progress towards settling an open question in both lists of ten open questions by Williamson and Shmoys~\cite{WS11}, and by Schuurman and Woeginger~\cite{SW99}.

The study of structural hardness of $k$-partite graphs is of independent interest, as it captures the intrinsic hardness for a large family of scheduling problems. Other than the ones already mentioned, this generalisation also implies tight inapproximability to the problem of minimising the weighted completion time for  precedence-constrained jobs on a single machine, and the problem of minimising the makespan of precedence-constrained jobs on identical parallel machine, and hence unifying the results of Bansal and Khot\cite{BK09} and Svensson~\cite{Ola11}, respectively.
\end{abstract}

%\medskip
%\noindent

\setcounter{page}{1}

%\bibliography{/home/ricoz/mainsync/research/literature/bibliographies/literature}
\newcommand{\Pprecpmtnn}{\textsc{P}$|$\textsf{prec,\xspace pmtn}$\xspace\xspace|$\textsc{C}${}_{max}$\xspace}
\section{Introduction}
The study of scheduling problems is motivated by the natural need to efficiently allocate limited resources over the course of time. While some scheduling problems can be solved to optimality in polynomial time, others turn out to be NP-hard. This difference in computational complexity can be altered by many factors, from the machines model that we adopt, to the requirements imposed on the jobs, as well as the optimality criterion of a feasible schedule. For instance, if we are interested in minimising the completion time of the latest job in a schedule (known as the maximum makespan), then the scheduling problem is NP-hard to approximate within a factor of $3/2-\epsilon$, for any $\epsilon>0$, if the machines are \emph{unrelated}, whereas it admits a Polynomial Time Approximation Scheme (PTAS) for the case of \emph{identical parallel} machines~\cite{HS87}. Adopting a model in between the two, in which the machines run at different speeds, but do so uniformly for all jobs (known as \emph{uniform parallel} machines), also leads to a PTAS for the scheduling problem~\cite{HS88}.

Although this somehow suggests a similarity in the complexity of scheduling problems between identical parallel machines and uniform parallel machines, our hopes for comparably performing algorithms seem to be shattered as soon as we add precedence requirements among the jobs. 
On the one hand, we know how to obtain a $2$-approximation algorithm for the problem where the parallel machines are identical~\cite{Gra66,GR08} (denoted as \Pprec in the language of~\cite{GLLR79}), whereas on the other hand the best approximation algorithm known to date for the uniform parallel machines case (denoted as \Qprec), gives a $\log(m)$-approximation guarantee~\cite{CS99,CB01}, $m$ being the number of machines. In fact obtaining a constant factor approximation algorithm for the latter, or ruling out any such result is a major open problem in the area of scheduling algorithms. Perhaps as a testament to that, is the fact that it is listed by Williamson and Shmoys~\cite{WS11} as Open Problem 8, and by Schuurman and Woeginger~\cite{SW99} as Open Problem 1.

Moreover, our understanding of scheduling problems even on the same model of machines does not seem to be complete either. On the positive side, it is easy to see that the maximum makespan of  any feasible schedule for \Pprec is at least $\max \left\{ L, n/m\right\}$, where $L$ is the length of the longest chain of precedence constraints in our instance, and $n$ and $m$ are the number of jobs and machines respectively. The same lower bound still holds when we allow preemption, i.e., the scheduling problem \Pprecpmtnn. 
%By preemption we mean the execution of a job can be stopped and resumed later, and not necessarily on the same machine. 
Given that both 2-approximation algorithms of ~\cite{Gra66} and~\cite{GR08} rely in their analysis on the aforementioned lower bound, then they also yield a 2-approximation algorithm for \Pprecpmtnn. However, on the negative side, our understanding for \Pprecpmtnn is much less complete. For instance,  we know that it is NP-hard to approximate \Pprec within any constant factor strictly better than $4/3$~\cite{LK79}, and assuming  (a variant of) the \uniquegames Conjecture, the latter lower bound is improved to 2~\cite{Ola11}. However for \Pprecpmtnn, only NP-hardness is known. It is important to note here that the hard instances yielding the $(2-\epsilon)$ hardness for \Pprec are  easy instances for \Pprecpmtnn. Informally speaking, the hard instances for \Pprec can be thought of as $k$-partite graphs, where each partition has $n+1$ vertices that correspond to $n+1$ jobs, and the edges from a layer to the layer above it emulate the precedence constraints. The goal is to schedule these $(n+1)k$ jobs on $n$ machines. If the $k$-partite graph is \emph{complete}, then any feasible schedule has a makespan of at least $2k$, whereas if the graph was a collection of \emph{perfect matchings} between each two consecutive layers, then there exists a schedule whose makespan is $k+1$\footnote{In fact, the gap is between $k$-partite graphs that have nice structural properties in the completeness case,  and behave like node expanders in the soundness case.}. However, if we allow preemption, then it is easy to see that even if the $k$-partite graph is complete, one can nonetheless find a feasible schedule whose makespan is $k+1$.

The effort of closing the inapproximability gap between the best approximation guarantee and the best known hardness result for some scheduling problems was successful in recent years; two of the results that are of particular interest for us are~\cite{BK09} and~\cite{Ola11}. Namely, Bansal and Khot studied in~\cite{BK09} the scheduling problem  \Oprec, the problem of scheduling precedence constrained jobs on a single machine, with the goal of minimsing the weighted sum of completion time, and proved tight inapproximability results for it, assuming a variant of the \uniquegames Conjecture. Similarly, Svensson proved in~\cite{Ola11} a hardness of $2-\epsilon$ for \Pprec, assuming the same conjecture. In fact, both papers relied on a structural hardness result for bipartite graphs, first introduced in~\cite{BK09}, by reducing a bipartite graph to a scheduling instance which leads to the desired hardness factor. 

\paragraph{Our results} We propose a natural but non-trivial generalisation of the structural hardness result of~\cite{BK09} from bipartite to $k$-partite graphs, that captures the intrinsic hardness of a large family of scheduling problems. Concretely, this generalisation yields \begin{enumerate}
\item A super constant hardness for \Qprec, making progress towards resolving an open question by~\cite{WS11,SW99} 
\item A hardness of $2-\epsilon$ for \Pprecpmtnn, even for the case where the processing time of each jobs is 1, denote by \Pprecpmtn, and hence closing the gap for this problem.
\end{enumerate}
Also, the results of~\cite{BK09} and~\cite{Ola11} will still hold for \Oprec and \Pprec, respectively, under the same assumption.
%We make progress towards understanding the inapproximability of both \Qprec and \Pprecpmtn scheduling problems, by observing a tight relation between the aforementioned problems, and the problem of detecting whether $k$-partite graphs have  a certain \emph{structure} or not. For the case of bipartite graphs, i.e., $k=2$,  (a variant of) this structure was used by Bansal and Khot \cite{BK09}, and Svensson~\cite{Ola11} to prove tight inapproximability results for \Oprec and \Pprec respectively. Assuming a structural hardness result for $k$-partite graphs, formalised in Hypothesis~\ref{hyp:kpartite}, we prove a super constant hardness factor for \Qprec. In other words , 

On the one hand, our generalisation rules out any constant factor polynomial time approximation algorithm for the scheduling problem \Qprec. On the other hand, one may speculate that the preemption flexibility when added to the scheduling problem \Pprec may render this problem easier, especially that the hard instances of the latter problem become easy when preemption is allowed. Contrary to such speculations, our generalisation to $k$-partite graphs enables us to prove that it is NP-hard to approximate the scheduling problem \Pprecpmtn within any factor strictly better than 2. Formally, we prove the following: \begin{theorem}
\label{thm:main2} Assuming Hypothesis~\ref{hyp:kpartite}, it is NP-hard to approximate the scheduling problems \Pprecpmtn within any constant factor strictly better than 2, and \Qprec within any constant factor.
\end{theorem}
This suggests that the intrinsic hardness of a large family of scheduling problems seems to be captured by structural hardness results for $k$-partite graphs. For the case of $k=2$, our hypothesis coincides with the structure bipartite hardness result of~\cite{BK09}, and  yields the following result: 
\begin{theorem}
\label{thm:main1}
Assuming a variant of the \uniquegames Conjecture, it is NP-hard to approximate the scheduling problem \Pprecpmtn within any constant factor strictly less than 3/2.
\end{theorem}
In fact, the $3/2$ lower bound holds even if we only assume that \Oprec is NP-hard to approximate within any factor strictly better than 2, by noting the connection between the latter and a certain bipartite ordering problem. This connection was observed and used by Svensson \cite{Ola11} to prove tight hardness of approximation lower bounds for \Pprec, and this yields a somehow stronger statement; even if the \uniquegames Conjecture turns out to be false, \Oprec might still be hard to approximate to within a factor of $2-\epsilon$, and our result for \Pprecpmtn will still hold as well. Formally, 
 \begin{corollary}
For any $\epsilon >0$, and $\eta \geq \eta(\epsilon)$, where $\eta(\epsilon)$ tends to 0 as $\epsilon$ tends to 0, if \Oprec has no $(2-\epsilon)$-approximation algorithm, then \Pprecpmtn has no $(3/2 - \eta)$-approximation algorithm.
\end{corollary}

Although we believe that Hypothesis~\ref{hyp:kpartite} holds, the proof is still eluding us. Nonetheless, understanding the structure of $k$-partite graphs seems to be a very promising direction to understanding the inapproximability of scheduling problems,  due to its manifold implications on the latter problems. As mentioned earlier, a similar structure for bipartite graphs was proved assuming a variant of the \uniquegames Conjecture in~\cite{BK09} (see Theorem~\ref{thm:BK}), and we show in Section~\ref{sec:PerfectMatching} how to extend it to $k$-partite graphs, while maintaining a somehow similar structure. However the resulting structure does not suffice for our purposes, i.e., does not satisfy the requirement for Hypothesis~\ref{hyp:kpartite}. Informally speaking, a bipartite graph corresponding to the completeness case of Theorem~\ref{thm:BK}, despite having a \emph{nice} structure, contains some noisy components that we cannot fully control. This follows from the fact that these graphs are derived from \uniquegames \emph{PCP-like} tests, where the resulting noise is either intrinsic to the \uniquegames instance (i.e., from the non-perfect completeness of the \uniquegames instance), or artificially added by the test. Although we can overcome the latter, the former prohibits us from replicating the structure of the bipartite graph to get a $k$-partite graph with an \emph{equally nice} structure. 

\paragraph{Further Related Work } The scheduling problem \Pprecpmtn was first shown to be NP-hard by Ullman \cite{Ull76}. However, if we drop the precedence rule, the problem can be solved to optimality in polynomial time \cite{Mcn59}. Similarly, if the precedence constraint graph is a \emph{tree}\cite{MC69,MC70,GJ80} or the number of machines is 2 \cite{MC69,MC70}, the problem also becomes solvable in polynomial time. Yet,  for an arbitrary precedence constraints structure, it remains open whether the problem is polynomial time solvable when the number of machines is a constant greater than or equal to 3 \cite{WS11}. A closely related problem to \Pprecpmtnn is \Pprec, in which preemption is not allowed. In fact the best 2-approximation algorithms known to date for \Pprecpmtnn were originally designed to approximate \Pprec~\cite{Gra66,GR08}, by noting the common lower bound for a makespan to any feasible schedule for both problems. As mentioned earlier, ~\cite{LK79} and~\cite{Ola11} prove a $4/3-\epsilon$ NP-hardness, and $2-\epsilon$ UGC-hardness respectively for \Pprec, for any $\epsilon>0$. However, to this date, only NP-hardness is known for the \Pprecpmtn scheduling problem. Although one may speculate that allowing preemption might enable us to get better approximation guarantees, no substantial progress has been made in this direction since~\cite{Gra66} and~\cite{GR08}. 

One can easily see that the scheduling problem \Pprec is a special case of \Qprec, since it corresponds to the case where the speed of every machine is equal to 1, and hence the $(4/3-\epsilon)$ NP-hardness of~\cite{LK79} and the ($2-\epsilon$) UGC-hardness of~\cite{Ola11} also apply to \Qprec. Nonetheless, no constant factor approximation for this problems is known; a $\log(m)$-approximation algorithm was designed by Chudak and Shmoys \cite{CS99}, and Chekuri and Bender \cite{CB01} independently, where $m$ is the number of machines. %In fact reducing this gap for the scheduling problem \Qprec is highlighted as Open Problem 8 in the list of ten open problems in~\cite{WS11}, and as Open Problem 1 in~\cite{SW99}.

%\vspace{-0.13in}
\paragraph{Outline } We start in Section~\ref{sec:prelim} by defining the \uniquegames problem, along with the variant of the \uniquegames Conjecture introduced in~\cite{BK09}. We then state in Section~\ref{sec:skpartite} the structural hardness result for bipartite graphs proved in~\cite{BK09}, and propose our new hypothesis for $k$-partite graphs (Hypothesis~\ref{hyp:kpartite}) that will play an essential role in the hardness proofs of Section~\ref{sec:lb}. Namely, we use it in Section~\ref{sec:speed} to prove a super constant inapproximability result for the scheduling problem \Qprec, and $2-\epsilon$ inapproximability for \Pprecpmtn. The reduction for the latter problem can be seen as replicating a certain scheduling instance $k-1$ times, and hence we note that if we settle for one copy of the instance, we can prove an inapproximability of $3/2$, assuming the variant of the \uniquegames Conjecture of \cite{BK09}. In Section~\ref{sec:PerfectMatching}, we prove a structural hardness result for $k$-partite graphs which is similar to Hypothesis~\ref{hyp:kpartite}, although not sufficient for our scheduling problems of interest. We note in Section~\ref{sec:LP} that the integrality gap instances for  the natural Linear Programming (LP) relaxation for \Pprecpmtn, have a very similar structure to the instances yielding the hardness result.
%As noted earlier, the Linear Program's integrality gap is done in two steps, to highlight the similarity between the structure of the family of the $3/2$-integrality gap instances, and that of the hard instances yielding the $3/2$-hardness result. 

%We then resort in Section~\ref{sec:kpartite} to a new hypothesis concerning $k$-partite graphs, to prove tight hardness results for the scheduling problem \Pprecpmtn, and rule out any constant factor approximation for \Qprec. As an evidence for the plausibility of such structure on $k$-partite graphs, we present in Section~\ref{sec:motivation} two promising approaches for extending structural hardness results from bipartite graphs to $k$-partite graphs that might yield the desired result. We also present a reduction from a \dtoone instance to a bipartite graph to prove a different structural hardness result for bipartite graphs that avoids the uncontrolled noise mentioned earlier. We conclude in Section~\ref{sec:conc} with further possible directions to close the gap for this family of scheduling problems. 

\section{Preliminaries}
\label{sec:prelim}
In this section, we start by introducing the \uniquegames problem, along with a variant of Khot's \uniquegames conjecture as it appears in \cite{BK09}, and then we formally define the scheduling problems of interest.%, and define the \dtooneg problem. We also suggest a variant to the \dtoone Conjecture.
%\subsection{\uniquegames}
%We restrict ourselves to \uniquegames instances defined over bipartite graphs. 
\begin{definition}
A \uniquegames instance $\mc{U}(G=(V,W,E), [R], \Pi)$ is defined by a bipartite graph $G = (V,W,E)$ with bipartitions $V$ and $W$ respectively, and edge set $E$. Every edge $(v,w)\in E$ is associated with a bijection map $\pi_{v,w}\in \Pi$ such that $\pi_{v,w}:[R] \mapsto [R]$, where $[R]$ is the label set. The goal of this problem is find a labeling $\Lambda: V\cup W \mapsto [R]$ that maximises the number of satisfied edges in $E$, where an edge $(u,v) \in E$ is satisfied by $\Lambda$ if $\pi_{v,w}(\Lambda(w)) = \Lambda(v)$. 
\end{definition}
Bansal and Khot~\cite{BK09} proposed the variant of the \uniquegames Conjecture in Hypothesis~\ref{hyp:bk09}, and used it to (implicitly) prove the structural hardness result for bipartite graphs in Theorem~\ref{thm:BK}.
%We state a variant of the \uniquegames conjecture, due to Bansal and Khot \cite{BK09}. %In fact, the construction we use in Section~\ref{sec:app1} is the same as the construction used in \cite{BK09} to prove a hardness of $2-\epsilon$ for the \Oprec scheduling problem. % (where $\sum_{j\in \mc{J}}w_j C_j$ is the weighted sum of completion times).
\begin{hypothesis}
\label{hyp:bk09}
{\normalfont [Variant of the UGC\cite{BK09}]} For arbitrarily small constants $\eta, \zeta, \delta >0$, there exists an integer $R = R(\eta,\zeta,\delta)$ such that for a \uniquegames instance $\mc{U}(G=(V,W,E),[R],\Pi)$, it is NP-hard to distinguish between:\begin{itemize}
\item (YES Case: ) There are sets $V' \subseteq V$, $W' \subseteq W$ such that $|V'| \geq (1-\eta) |V|$ and $|W'| \geq (1-\eta)|W|$, and a labeling $\Lambda:V\cup W \mapsto [R]$ such that all the edges between the sets $(V',W')$ are satisfied.
\item (NO Case: ) No labeling to $\mc{U}$ satisfies even a $\zeta$ fraction of edges. Moreover, the instance satisfies the following expansion property. For every $S\subseteq V$, $T\subseteq W$, $|S| = \delta |V|$, $|T| = \delta |W|$, there is an edge between $S$ and $T$.% we have $\Gamma(S) \geq (1-\delta)|W|$, where $\Gamma(S) := \{w\in W: \exists v\in S, \,\, (v,w) \in E \}.$
\end{itemize}
\end{hypothesis}
%The following structural hardness results for bipartite graphs is also implicit in \cite{BK09}: 
\begin{theorem}{\normalfont[Section 7.2 in \cite{BK09}]}
\label{thm:BK}
For every $\epsilon, \delta >0$, and positive integer $Q$, the following problem is NP-hard assuming Hypothesis~\ref{hyp:bk09}: given an n-by-n bipartite graph $G = (V,W,E)$,  distinguish  between the following two cases: \begin{itemize}
\item YES Case: $V$ can be partitioned into $V_{0}, \dots, V_{Q-1}$ and $W$ can be partitioned into $W_0,\dots, W_{Q-1}$, such that
\begin{itemize}
\item There is no edge between $V_{i}$ and $W_{j}$ for all $0 \leq j < i < Q$.
\item $|V_{i}| \geq \frac{(1-\epsilon)}{Q}n$ and $|W_{i}| \geq \frac{(1-\epsilon)}{Q}n$, for all $i\in[Q]$.
\end{itemize}%\item YES Case: We can partition $V$ into disjoints sets $V_1,V_2,\dots,V_Q$ with $|V_1|,|V_2|,\dots,|V_Q| \geq \frac{1-\epsilon}{Q}$, and $W$ into disjoint sets $W_1,W_2,\dots,W_Q$ with $|W_1|, |W_2|,\dots, |W_Q|  \geq \frac{1-\epsilon}{Q}$ such that for every $i \in \{1,2,\dots,Q\}$ and any vertex $w \in W_i$, $w$ only have edges to vertices in $V_j$ for $j \leq i$.
\item NO Case: For any $S\subseteq V$, $T \subseteq W$, $|S| = \delta n$, $|T| = \delta n$, there is an edge between $S$ and $T$.
\end{itemize}
\end{theorem}
In the scheduling problems that we consider, we are given a set $\mc{M}$ of machines and a set $\mc{J}$ of jobs with precedence constraints, and the goal is find a feasible schedule in a way to minimise the makespan, i.e., the maximum completion time. We will be interested in the following two variants of this general setting: \begin{description}
\item[\Pprecpmtnnn: ] In this model, the machines are assumed to be be parallel and identical, i.e., the processing time of a job $J_j \in \mc{J}$ is the same on any machine $M_i \in \mc{M}$ ($p_{i,j} = p_j$ for all $M_i \in \mc{M}$). Furthermore, preemption is allowed, and hence the processing of a job can be paused and resumed at later stages, not necessarily on the same machine.
\item[\Qprec: ] In this model, the machines are assumed to be parallel and uniform, i.e., each machine $M_i \in \mc{M}$ has a speed $s_i$, and the time it takes to process job $J_j \in \mc{J}$ on this machine is $p_j/s_i$.
\end{description}
Before we proceed we give the following notations that will come in handy in the remaining sections of the paper. For a positive integer $Q$, $[Q]$ denotes the set $\{0,1,\dots,Q-1\}$. In a scheduling context, we say that a job $J_i$ is a predecessor of a job $J_j$, and write it $J_i \prec J_j$, if in any feasible schedule, $J_j$ cannot start executing before the completion of job $J_i$. Similarly, for two \emph{sets} of jobs $\mc{J}_i$ and $\mc{J}_j$, $\mc{J}_i \prec \mc{J}_j$ is equivalent to saying that all the jobs in $\mc{J}_j$ are successors of all the jobs in $\mc{J}_i$.
\section{Structured $k$-partite Problem}
\label{sec:skpartite}
We propose in this section a natural but nontrivial generalisation of Theorem~\ref{thm:BK} to $k$-partite graphs.  Assuming hardness of this problem, we can get the following hardness of approximation results:
\begin{enumerate}
\item It is NP-hard to approximate \Qprec within any constant factor.
\item It is NP-hard to approximate \Pprecpmtn within a $2-\epsilon$ factor.
\item It is NP-hard to approximate \Oprec within a $2-\epsilon$ factor.
\item It is NP-hard to approximate \Pprec within a $2-\epsilon$ factor.
\end{enumerate}
The first and second result are presented in Section~\ref{sec:speed} and~\ref{sec:Pprecpmtn}, respectively. Moreover, one can see that the reduction presented in~\cite{BK09} for the scheduling problem \Oprec holds using the hypothesis for the case that $k=2$. The same holds for the reduction in~\cite{Ola11}  for the scheduling problem \Pprec. This suggests that this structured hardness result for $k$-partite graphs somehow unifies a large family of scheduling problems, and captures their common intrinsic hard structure.% These results give us the intuition that we can attack variant scheduling problems by attacking this problem. Let us state our new hypothesis now.
\begin{hypothesis}\label{hyp:kpartite}[$k$-partite Problem]
For every $\epsilon, \delta >0$, and constant integers $k, Q>1$, the following problem is NP-hard: given a $k$-partite graph $G= ( V_1, ...,$ $V_k , E_1,..., E_{k-1} )$ with $|V_i|=n$ for all $1 \leq i \leq k$ and $E_i$ being the set of edges between $V_i$ and $V_{i+1}$ for all $1 \leq i < k$, distinguish between following two cases:
%For every $\epsilon, \delta >0$, and constant integers $k\geq 2, Q>1$, there exists an $n=n(\epsilon, \delta, Q, k)$, such that given a $k$-partite graph $G= ( V_1, ..., V_k , E_1,..., E_{k-1} )$ with $|V_i|=n$  and $E_i$ is a set of edges between $V_i$ and $V_{i+1}$ , such that it is NP-hard to distinguish between following two cases:

\begin{itemize}
\item YES Case: every $V_i$ can be partitioned into $V_{i,0}, ..., V_{i,Q-1}$, such that
\begin{itemize}
\item There is no edge between $V_{i,j_1}$ and $V_{i-1,j_2}$ for all $1 < i \leq k, j_1 < j_2 \in [Q]$.

\item $|V_{i,j}| \geq \frac{(1-\epsilon)}{Q}n$, for all $1 \leq i \leq k,  j \in [Q]$.
\end{itemize}
\item NO Case: For any $1 < i \leq k$ and any two sets $S\subseteq V_{i-1}$, $T \subseteq V_{i}$, $|S|=\delta n$, $|T|=\delta n$, there is an edge between $S$ and $T$.
\end{itemize}
\end{hypothesis}
This says that if the $k$-partite graph $G=(V_1, ..., V_k, E_1,..., E_{k-1})$ satisfies the YES Case, then for every $1 \leq i \leq k-1$, the induced subgraph $\tilde{G}=(V_i,V_{i+1},E_i)$ behaves like the YES Case of Theorem~\ref{thm:BK}, and otherwise, every such induced subgraph corresponds to the NO case. Moreover, if we think of $G$ as a directed graph such that the edges are oriented from $V_i$ to $V_{i-1}$, then all the partitions in the YES case are consistent in the sense that a vertex $v \in V_{i,j}$ can only have paths to vertices $v' \in V_{i',j'}$ if $i'<i \leq k$ and $j' \leq j \leq Q-1$.  

We can prove that assuming the previously stated variant of the \uniquegames Conjecture, Hypothesis~\ref{hyp:kpartite} holds for $k=2$. Also we can extend Theorem~\ref{thm:BK} to a $k$-partite graph using a perfect matching approach which results in the following theorem. We delegate its proof to Appendix~\ref{sec:PerfectMatching}.
\begin{theorem}
\label{thm:BKkpartite}
%For every $\epsilon, \delta >0$, and constant integers $k\geq 2, Q>1$, there exists an $n=n(\epsilon, \delta, Q, k)$, such that given a $k$-partite graph $G=(V_1, \dots, V_k, E_1,\dots, E_{k-1})$ with $|V_i|=n$  and $E_i$ is a set of edges between $V_i$ and $V_{i+1}$ , such that assuming (a variant of) \uniquegames Conjecture~\cite{BK09} it is NP-hard to distinguish between following two cases:
For every $\epsilon, \delta >0$, and constant integers $k,Q>1$, the following problem is NP-hard: given a $k$-partite graph $G=(V_1, \dots, V_k, E_1,\dots, E_{k-1})$ with $|V_i|=n$  and $E_i$ being the set of edges between $V_i$ and $V_{i+1}$ , distinguish between following two cases:

\begin{itemize}
\item YES Case: every $V_i$ can be partitioned in to $V_{i,0}, ..., V_{i,Q-1},V_{i,err}$, such that
\begin{itemize}
\item There is no edge between $V_{i,j_1}$ and $V_{i-1,j_2}$ for all $1 < i \leq k, j_1 \neq j_2 \in [Q]$.
\item $|V_{i,j}| \geq \frac{(1-\epsilon)}{Q}n$ for all $1 \leq i \leq k,  j \in [Q]$.
\end{itemize}
\item NO Case: For any $1 < i \leq k$ and any two sets $S\subseteq V_i$, $T \subseteq V_{i-1}$, $|S|=\delta n$, $|T|=\delta n$, there is an edge between $S$ and $T$.
\end{itemize}
\end{theorem} 

Note that in the YES Case, the induced subgraphs on $\{V_{i,j}\}$ for $1\leq i \leq k$, $0\leq j \leq Q-1$, have the perfect structure that we need for our reductions to scheduling problems. However, we do not get the required structure between the noise partitions (i.e., $\{V_{i,err}\}$ for $1\leq i \leq k$), which will prohibit us from getting the desired gap between the YES and NO Cases when performing a reduction from this graph to our scheduling instances of interest. The structure of the noise that we want is that the vertices in the noise partition are only connected to the vertices in the noise partition of the next layer. %The proof of Theorem~\ref{thm:BKkpartite} is presented in Appendix~\ref{sec:PerfectMatching}.

\section{Lower Bounds for Scheduling Problems}
\label{sec:lb}
In this section, we show that, assuming Hypothesis~\ref{hyp:kpartite},  there is no constant factor approximation algorithm for the scheduling problem \Qprec , and there is no $c$-approximation algorithm for the scheduling problem \Pprecpmtn, for any constant $c$ strictly better than 2. We also show that, assuming a special case of Hypothesis~\ref{hyp:kpartite}, i.e., $k=2$ which is equivalent to (a variant) of \uniquegames Conjecture (Hypothesis~\ref{hyp:bk09}), there is no approximation algorithm better than $3/2-\epsilon$ for \Pprecpmtn, for any $\epsilon>0$.
\subsection{\Qprec}
\label{sec:speed} 
%\subsubsection{Overview}
 In this section, we reduce a given $k$-partite graph $G$ to an instance $\mc{I}(k)$ of the scheduling problem \Qprec, and show that if $G$ corresponds to the YES Case of Hypothesis~\ref{hyp:kpartite}, then the maximum makespan of $\mc{I}(k)$ is roughly $n$, whereas a graph corresponding to the NO Case leads to a scheduling instance whose makespan is roughly the number of vertices in the graph, i.e., $nk$. Formally, we prove the following theorem.
\begin{theorem}
\label{thm:kpartite2}
Assuming Hypothesis~\ref{hyp:kpartite}, it is NP-hard to approximate the scheduling problem \Qprec within any constant factor.
\end{theorem}
\subsubsection{Reduction}	
We present a reduction from a $k$-partite graph $G=(V_1, ..., V_k, E_1,...,$ $ E_{k-1})$ to an instance $\mathcal{I}(k)$ of the scheduling problem \Qprec. The reduction is parametrised by a constant $k$, a constant $Q\gg k$ such that $Q$ divides $n$, and a large enough value $m \gg nk$. 
\begin{itemize}
\item For each vertex in $v \in V_{i}$, let $\mc{J}_{v,i}$ be a set of $m^{2(k-i)}$ jobs with processing time $m^{i-1}$, for every $ 1 \leq i \leq k$.
\item For each edge $e = (v,w) \in E_i$, we have $\mc{J}_{v,i} \prec \mc{J}_{w,i+1}$, for $1 \leq i < k$ .
\item For each $1 \leq i \leq k$ we create a set $\mc{M}_i$ of $m^{2(k-i)}$ machines with speed $m^{i-1}$.
\end{itemize}
%W.l.o.g. we assume that $n$ is divisible by $Q$, and we let $m>>nk$.
\subsubsection{Completeness}
We show that if the given $k$-partite graph satisfies the properties of the YES Case, then there exist a schedule with makespan $(1+\epsilon_1) n$ for some small $\epsilon_1>0$. Towards this end, assume that the given $k$-partite graph satisfies the properties of the YES Case and let $\{V_{i,j}\}$ for $1\leq i \leq k$ and $0\leq j \leq Q-1$ be the claimed partitioning of Hypothesis~\ref{hyp:kpartite}. 

The partitioning of the vertices naturally induces a partitioning $\{\tilde{\mc{J}}_{i,j}\}$ for the jobs for $1\leq i \leq k$ and $0 \leq j \leq Q-1$ in the following way: \begin{align*}
\tilde{\mc{J}}_{i,j} = \bigcup_{v\in V_{i,j}}\mc{J}_{v,i}
\end{align*}

Consider the schedule where for each $1\leq i \leq k$, all the jobs in a set $\tilde{\mc{J}}_{i,0},\dots,\tilde{\mc{J}}_{i,Q-1}$ are scheduled on the machines in $\mc{M}_i$. Moreover, we start the jobs in $\tilde{\mc{J}}_{i,j}$ after finishing the jobs in both $\tilde{\mc{J}}_{i-1,j}$ and $\tilde{\mc{J}}_{i,j-1}$ (if such sets exist). In other words, we schedule the jobs as follows (see Figure~\ref{fig:difspeed}):
 \begin{itemize}
\item For each $1\leq i \leq k$, we first schedule the jobs in $\tilde{\mc{J}}_{i,0}$, then those in $\tilde{\mc{J}}_{i,1}$ and so on up until $\tilde{\mc{J}}_{i,Q-1}$. The scheduling of the jobs on machines in $\mc{M}_0$ starts at time 0 in the previously defined order.
\item For each $2 \leq i \leq k$, we start the scheduling of jobs $\tilde{\mc{J}}_{i,0}$ right after the completion of the jobs in $\tilde{\mc{J}}_{i-1,0}$.
\item To respect the remaining precedence requirements, we start scheduling the jobs in $\tilde{\mc{J}}_{i,j}$ right after the execution of jobs in $\tilde{\mc{J}}_{i,j-1}$ and as soon as the jobs in $\tilde{\mc{J}}_{i-1,j}$ have finished executing, for $2\leq i \leq k$ and $1 \leq j \leq Q-1$. 
\end{itemize} 
By the aforementioned construction of the schedule, we know that the precedence constraints are satisfied, and hence the schedule is feasible. That is, since we are in YES Case, we know that vertices in $V_{i',j'}$ might only have edges to the vertices in $V_{i,j}$ for all $1 \leq i' < i \leq k$ and $1 \leq j' \leq j < Q$, which means that the precedence constraints may only be from the jobs in $\tilde{\mc{J}}_{i',j'}$ to jobs in  $\tilde{\mc{J}}_{i,j}$  for all $1 \leq i' < i \leq k$ and $0 \leq j' \leq j < Q$. Therefore the precedence constraints are satisfied. 

Moreover, we know that there are at most $m^{2(k-i)}n(1+\epsilon)/Q$ jobs of length $m^{i-1}$ in $\tilde{\mc{J}}_{i,j}$, and $m^{2(k-i)}$ machines with speed $m^{i-1}$ in each $\mc{M}_i$ for all $1\leq i\leq k$, $j\in[Q]$. This gives that it takes $(1+\epsilon)n/Q$ time to schedule all the jobs in $\tilde{\mc{J}}_{i,j}$ on the machines in $\mc{M}_i$ for all $1\leq i\leq k$, $j\in[Q]$, which in turn implies that we can schedule all the jobs in a set $\tilde{\mc{J}}_{i,j}$ between time $(i+j-1)(1+\epsilon)n/Q$ and $(i+j)(1+\epsilon)n/Q$. This gives that the makespan is at most $(k+Q)(1+\epsilon)n/Q$ which is equal to $(1+\epsilon_1)n$, by the assumption that $Q \gg k$.  

% $\mc{J}_{v,i}$ At time zero, we start the jobs according to vertices in $V_{11}$ on the $M_1$, i.e., $\cup_{v \in V_{11}} J_{v1}$. Then we continue scheduling the jobs in $\cup_{v \in V_{1i}} J_{vi}$ for $ 1 \leq i < Q$ in increasing order of $i$. First partition finishes before time $(1+\epsilon)n/Q$, this enable us to start jobs in  $\cup_{v \in V_{2i}} J_{v2}$ on machines in $M_2$. Similarly, we schedule the jobs in $\cup_{v \in V_{ji}} J_{vi}$ on machine $j$ from time $(i+j-1)(1+\epsilon)n/Q$ to $(i+j)(1+\epsilon)n/Q$. This gives that the makespan is at most $(k+Q)(1+\epsilon)n/Q$ which is equal to $(1+\epsilon_1)n$ assuming that $Q>>k$.  
\subsubsection{Soundness}
We shall now show that if the $k$-partite graph $G$ corresponds to the NO Case of Hypothesis~\ref{hyp:kpartite}, then any feasible schedule for $\mc{I}(k)$ must have a makespan of at least $cnk$, where $c:=(1-2\delta)(1-k^2/m)$ can be made arbitrary close to one. %This will follow from the following Lemma.

%Assume that we are given a valid schedule $\sigma$  for $\mc{I}(k)$ of makespan at most $nk$, then the following is true.
\begin{lemma}
\label{lem:sched}
In a feasible schedule $\sigma$ for $\mc{I}(k)$ such that the makespan of $\sigma$ is at most $nk$, the following is true: for every $1\leq i \leq k$, at least a $(1-k^2/m)$ fraction of the jobs in $\mc{L}_i =\cup_{v \in V_i} \mc{J}_{v,i}$ are scheduled on machines in $\mc{M}_i$.
%For any small constant $\gamma > 0$, at least $1-\gamma$ fraction of jobs in $\mc{L}_i =\cup_{v \in V_i} \mc{J}_{v,i}$ is scheduled on machines in $\mc{M}_i$ for all $1 \leq i \leq k$.% in $\mc{S}$.
\end{lemma}
\begin{proof}
We first show that no job in $\mc{L}_i$ can be scheduled on machines in $\mc{M}_j$, for all $1 \leq j < i \leq k$. This is true, because any job $J\in \mc{J}_i$ has a processing time of $m^{i-1}$, whereas the speed of any machine $ M \in \mc{M}_j$ is $m^{j-1}$ by construction, and hence scheduling the job $J$ on the machine $M$ would require $m^{i-1}/m^{j-1} \geq m$ time steps. But since $m \gg nk$, this contradicts the assumption that the makespan is at most $nk$. %That is, the job length is $m^{i-1}$ and the speed on the machine is $m^{j-1}$. Therefore, it takes $m^{i-1}/m^{j-1} \geq m$ to finish the job which contradicts that the makespan is at most $nk$.

We now show that at most $k^2/m$ fraction of the jobs in $\mc{L}_i$ can be scheduled on the machines in $\mc{M}_j$ for $1\leq i < j \leq k$. Fix any such pair $i$ and $j$, and assume that all the machines in $\mc{M}_j$ process the jobs in $\mc{L}_i$ during all the $T \leq nk$ time steps of the schedule. This accounts for a total $T \frac{m^{2(k-j)}m^{j-1}}{m^{i-1}} \leq  m^{2k-j-i} nk $ jobs processed from $\mc{L}_i$, which constitutes at most $\frac{m^{2k-j-i} nk }{n m^{2(k-i)}} \leq \frac{k}{m}$ fraction of the total number of jobs in $\mc{L}_i$. 
%We also show that at most $k^2/m$ fraction of the jobs in $L_i$ can be done on the machines $M_j$ for $1 \leq i < j \leq k$. That is, for any $1 \leq i < j \leq k$, even if all the machines in $M_j$ works on jobs in $L_i$ hole the time, i.e. $nk$, they can only finish $nk\frac{m^{2(k-j)}m^{(j-1)}}{m^{i-1}}=m^{2k-j-i}nk$ jobs which is at most $k/m$ fraction of the jobs in $L_i$.  
\end{proof}
Let $\sigma$ be a schedule whose makespan is at most $nk$, and fix $\gamma>k^2/m$ to be a small constant. From Lemma~\ref{lem:sched} we know that for every $1\leq i \leq k$, at least an $(1-\gamma)$ fraction of the jobs in $\mc{L}_i$  is scheduled on machines  in $\mc{M}_i$. From the structure of the graph in the NO Case of the $k$-partite Problem, we know that we cannot start more than $\delta$ fraction of the jobs in $\mc{L}_i$ before finishing $(1-\delta)$ fraction of the jobs in $\mc{L}_{i-1}$, for all $2\leq i \leq k$. Hence the maximum makespan of any such schedule $\sigma$ is at least $(1-2\delta)(1-\gamma) nk$. See figure ~\ref{fig:difspeed}.

%Now it is we assume that $1-\gamma$ fraction of the jobs in $L_i$ in scheduled on machines in $M_i$ in schedule $\mc{S}$. Also form the properties of the NO case we know that we can not start more than $\delta$ fraction of job in $L_i$ before finishing at least $1-\delta$ fraction of the jobs in $L_{i-1}$ for all $1 < i \leq k$. This gives that the length of the schedule is at least $(1-2\delta)(1-\gamma)nk$. Comparing the makespan of YES case and NO case we get that for any constant $k$, there is no $k$ approximation algorithm for \Qprec. 

\subsection{\Pprecpmtn}
\label{sec:Pprecpmtn}
We present in this section a reduction from a $k$-partite graph to an instance of the scheduling problem \Pprecpmtn, and prove a tight inapproximability result for the latter, assuming Hypothesis~\ref{hyp:kpartite}. Formally, we prove the following result:
%The implication of the $k$-partite Problem~\ref{hyp:kpartite} turn out to be twofold: not only rules out any constant factor approximation for the scheduling problem \Qprec if true, but it also it imply tight inapproximability result for \Pprecpmtn.
\begin{theorem}
\label{thm:kpartite1}
Assuming Hypothesis~\ref{hyp:kpartite}, it is NP-hard to approximate the scheduling problem \Pprecpmtn within any constant factor strictly better than 2.
\end{theorem}
To prove this, we first reduce a $k$-partite graph $G=(V_1, ..., V_k, E_1,..., E_{k-1})$ to a scheduling instance $\tilde{\mc{I}}(k)$, and then show that \begin{enumerate}
\item If $G$ satisfies the YES Case of  Hypothesis~\ref{hyp:kpartite}, then $\tilde{\mc{I}}(k)$ has a feasible schedule whose makespan is roughly $kQ/2$.
\item if $G$ satisfies the NO Case of Hypothesis~\ref{hyp:kpartite}, then any schedule for $\tilde{\mc{I}}(k)$ must have a  makespan of roughly $kQ$.
\end{enumerate}

\subsubsection{Reduction}
\label{sec:kpartite:red}
The reduction has three parameters: an odd integer $k$, an integer $Q$ such that $Q \gg k$ and $n$ divides $Q$, and a real $  \epsilon \gg 1/Q^2 >0$. 

Given a $k$-partite graph $G=(V_1, ..., V_k, E_1,..., E_{k-1})$, we construct an instance $\tilde{\mc{I}}(k)$ of the scheduling problem \Pprecpmtn as follows:
 \begin{itemize}
\item For each vertex $v \in V_{2i-1}$ and every $ 1 \leq i \leq (k+1)/2$, we create a set $\mc{J}_{2i-1,v}$ of $Qn-(Q-1)$ jobs.
 \item For each vertex $v \in V_{2i}$ and every $ 1 \leq i < (k+1)/2$, we create a chain of length $Q-1$ of jobs, i.e., a set $\mc{J}_{2i,v}$ of $Q-1$ jobs
 \begin{align*}
\mc{J}_{2i,v} = \{J^1_{2i,v}, J^2_{2i,v},\dots,J^{Q-1}_{2i,v}\}  
\end{align*}
where we have $J^{l}_{2i,v} \prec J^{l+1}_{2i,v}$ for all $l \in \{1,2,\dots, Q-2\}$.
\item For each edge $e = (v,w) \in E_{2i-1}$ and every $1 \leq i < (k+1)/2$, we have $\mc{J}_{2i-1,v} \prec J^1_{2i,w}$.
\item For each edge $e = (v,w) \in E_{2i}$ and every $1 \leq i < (k+1)/2$, we have $J_{2i,v}^{Q-1} \prec \mc{J}_{2i+1,w}$.
%\item For each vertex $v \in V_{2i}$ and every $1 \leq i < (k+1)/2$, we have $J^{l}_{2i,v} \prec J^{l+1}_{2i,v}$ for all $l \in \{1,2,\dots, Q-2\}$ 
\end{itemize} 
Finally the number of machines is $(1+Q\epsilon)n^2$.  

Theorem~\ref{thm:kpartite1} now follows from the following lemma, whose proof can be found in Appendix~\ref{sec:appkprte}.
\begin{lemma}
\label{lem:prec}
Scheduling instance $\tilde{\mc{I}}(k)$ has the following two properties.
\begin{enumerate}
\item If $G$ satisfies the YES Case of Hypothesis~\ref{hyp:kpartite}, then $\tilde{\mc{I}}(k)$ has a feasible schedule whose makespan is $(1+\epsilon)kQ/2$, where $\epsilon$ can be arbitrary close to zero.
\item if $G$ satisfies the NO Case of Hypothesis~\ref{hyp:kpartite}, then any feasible schedule for $\tilde{\mc{I}}(k)$ must have a  makespan of $(1-\epsilon)kQ$, where $\epsilon$ can be arbitrary close to zero.
\end{enumerate}
\end{lemma}

Although not formally defined, one can devise a similar reduction for the case of $k=2$, and prove  a $3/2$-inapproximability result for \Pprecpmtn, assuming the variant of the \uniquegames Conjecture in~\cite{BK09}. We illustrate this in Appendix~\ref{UGLB} and prove the following result:
 \begin{theorem}
\label{thm:main22} For any $\epsilon>0$, it is NP-hard to approximate \Pprecpmtn within a factor of $3/2 -\epsilon$, assuming (a variant of) the \uniquegames Conjecture.%~\ref{hyp:bk09}.
\end{theorem}

\section{Discussion}
\label{sec:conc}
We proposed in this paper a natural but nontrivial generalisation of Theorem~\ref{thm:BK}, that seems to capture the hardness of a large family of scheduling problems with precedence constraints. It is interesting to investigate whether this generalisation also illustrates potential intrinsic hardness of other scheduling problems, for which the gap between the best known approximation algorithm and the best known hardness result persists.

On the other hand, a natural direction would be to prove Hypothesis~\ref{hyp:kpartite}; we show in Section~\ref{sec:PerfectMatching} how to prove a \emph{less-structured} version of it using the bipartite graph resulting from the variant of the \uniquegames Conjecture in~\cite{BK09}. One can also tweak the dictatorship $T_{\epsilon,t}$ of~\cite{BK09}, to yield a $k$-partite graph instead of a bipartite one. However, composing this test with a \uniquegames instance adds a noisy component to our $k$-partite graph, that we do not know how to control, since it is due to the non-perfect completeness of the \uniquegames instance. One can also try to impose (a variant of) this dictatorship test on \dtooneg instances, and perhaps prove the hypothesis assuming the \dtoone Conjecture, although we expect the size of the partitions to deteriorate as $k$ increases.   

\subsection*{Acknowledgments}
The authors are grateful to Ola Svensson for inspiring discussions and valuable comments that influenced this work. We also wish to thank Hyung Chan An, Laurent Feuilloley, Christos Kalaitzis and the anonymous reviewers for several useful comments on the exposition.
\bibliographystyle{abbrv}
\bibliography{lit}
\newpage 
\appendix
\newpage

\section{Proof of Lemma~\ref{lem:prec}}
\label{sec:appkprte}
In this section, we prove Lemma~\ref{lem:prec}, that is we show that the reduction in Section~\ref{sec:Pprecpmtn} from a k-partite graph $G$ to a scheduling instance $\mc{I}(k)$ yields a hardness of $2-\epsilon$ for the scheduling problem \Pprecpmtn, for any $\epsilon>0$. This follows from combining Lemmas~\ref{lem:lemma1} and~\ref{lem:lemma2}.

\begin{lemma} [{\bf Completeness}]
\label{lem:lemma1}
If the given $k$-partite graph $G$ satisfies the properties of the YES case of Hypothesis~\ref{hyp:kpartite}, then there exists a valid schedule for $\tilde{\mc{I}}(k)$ with maximum makespan $(1+\epsilon ')kQ/2$, where $\epsilon'$ can be made arbitrary close to zero. 
\end{lemma}
\begin{proof} Assume that $G$ satisfies the properties of the YES Case of Hypothesis~\ref{hyp:kpartite}, and let $\{V_{s,\ell}\}$ for $s \in \{1,\dots,k\}, \ell \in [Q]$ denote the \emph{good} partitioning of the vertices of $G$. We use this partitioning to derive a partitioning $\{\mc{S}_{i,j}\}$ for the jobs in the scheduling instance $\tilde{\mc{I}}(k)$ for $1\leq i \leq (k-1)Q/2-1$, $j \in [Q]$, where a set of jobs $\mc{S}_{i,j}$ can be either \emph{big} or \emph{small}. 

The intuition behind this big/small distinction is that a job $J$ is in a big set if it is part of the $Qn-(Q-1)$ copies of a vertex $v \in V_{2i-1}$ for $1\leq i \leq (k+1)/2$, and in a small set otherwise. %Let us first provide an intuitive schedule. In here by scheduling the vertices we mean the jobs corresponding to those vertices.

%In our schedule we can first schedule the jobs in corresponding to the vertices in $V_{1,0}$, i.e.,  \bigcup_{v \in V_{2i-1,j}}\mc{J}_{2i-1,v}. This allows us to schedule jobs is corresponding to the  $\mc{J}_{$
%As for the general partitioning criteria, two jobs $J_1,J_2$ satisfy $J_1 \in \mc{S}_{i,{j_1}}$ and $J_2 \in \mc{S}_{i,j_2}$, if both jobs are at a distance $i$ from their first predecessor in $\tilde{\mc{I}}(k)$. Similarly, a job $J$ is inside a set $\mc{S}_{i,j}$ if $J$ \emph{corresponds} to a vertex $v$ in $G$ such that $v \in V_{s,j}$ for some $s\in \{1,2,\dots,k\}$.
 These sets can now be formally defined as follows:
\begin{align*}
\text{Big sets: }& \mc{S}_{Q(i-1)+1,j} := \bigcup_{v \in V_{2i-1,j}}\mc{J}_{2i-1,v} && \forall 1 \leq i \leq \frac{k+1}{2}, j \in [Q]\\
\text{Small sets: }& \mc{S}_{Q(i-1)+1+l,j} := \bigcup_{v \in V_{2i,j}}J^l_{2i,v} && \forall 1 \leq i < \frac{k+1}{2}, j \in [Q], l \in[Q-1]
\end{align*}
We first provide a brief overview of the schedule before defining it formally. Since $\mc{S}_{1,0}$ is the set of the jobs corresponding to the vertices in $V_{1,0}$, scheduling all the jobs in $\mc{S}_{1,0}$ in the first time step enables us to start the jobs at the first layer of the chain corresponding to vertices in $V_{2,0}$ (i.e., $\mc{S}_{2,0}$). Therefore in the next time step we can schedule the jobs corresponding to the vertices in $V_{1,1}$, (i.e. $\mc{S}_{1,1}$) and $\mc{S}_{2,0}$. This further enables us to continue to schedule the jobs in the second layer of the chain corresponding to the vertices in  $V_{2,0}$ (i.e., $\mc{S}_{3,0}$), the jobs at the first layer of the chain corresponding to vertices in $V_{2,1}$ (i.e., $\mc{S}_{2,1}$), and the jobs corresponding to the vertices in $V_{1,2}$ (i.e., $\mc{S}_{1,2}$). We can keep going the same way, until we have scheduled all the jobs. Since the number of partitions of each vertex set $V_i$ is $Q$, and length of each of our chains is $Q-1$, we can see that in the suggested schedule, we are scheduling in each time step at most $Q$ sets, out of which exactly one is big, and none of the precedence constraints are violated (see Figure~\ref{fig:UGhard}).

Formally speaking, let $\mc{T}_t$ be the union of $\mc{S}_{i,j}$ such that $t=i+j-1$, %$T_t = \cup_{\{(i,j)|i+j=l, 1\leqei} W_{ij}$ for a
where $1\leq i\leq (k-1)Q/2+1$ and $j \in [Q]$, hence each $\mc{T}_t$ consist of at most $Q$ sets of the jobs in which exactly one of them is a big set and at most $Q-1$ of them are small sets. Therefore, for $t \in [(k+1)Q/2]$ we have
\begin{align*}
|\mc{T}_t| &\leq |V_{2i-1,j}|\cdot (Qn-(Q-1))+ |V_{2i,j}|\cdot (Q-1) \\
 & \leq (1/Q+\epsilon)n \cdot(Qn-(Q-1)) + (1/Q+\epsilon)n \cdot (Q-1)\\
 & \leq (1/Q+\epsilon)n\cdot(Qn) \leq (1+Q\epsilon)n^2
\end{align*}
One can easily see that all the jobs in a set $\mc{T}_t$ can be scheduled in a single time step since the number of machines is $(1+Q \epsilon)n^2$. Hence consider the following schedule: for each $t \in [(k+1)Q/2]$, schedule all the jobs in $\mc{T}_t$ between time $t$ and $t+1$. We claim that this schedule does not violate any precedence constraint. This is true because we first schedule the predecessors of the job, and then the job in the following steps. Formally, if $J_1 \prec J_2$ with $J_1 \in \mc{T}_{t_1}$ and $J_2 \in \mc{T}_{t_2}$, then $t_1 < t_2$. The structure of such schedule is depicted in Figure~\ref{fig:kpartitepmtn}. 
\end{proof}

\begin{lemma} [{\bf Soundness}]
\label{lem:lemma2}
If the given $k$-partite graph $G$ satisfies the properties of the No Case of Hypothesis~\ref{hyp:kpartite}, then any feasible schedule for $\tilde{\mc{I}}(k)$ has a maximum makespan of at least $(1-\epsilon ')kQ$, where $\epsilon'$ can be made arbitrary close to zero. 
\end{lemma}
\begin{proof}
Assume that $G$ satisfies the NO Case of Hypothesis~\ref{hyp:kpartite}, and consider the following partitioning of the jobs: 
\begin{align*}
\text{Big partitions: }\mc{S}_{Q(i-1)+1} := \cup_{v \in V_{2i-1}}\mc{J}_{2i-1,v} && \forall 1 \leq i \leq (k+1)/2 \\
\text{Small partitions: }\mc{S}_{Q(i-1)+1+l} := \cup_{v \in V_{2i}}J^l_{2i,v} && \forall 1 \leq i < (k+1)/2, l \in[Q-1]\\
\end{align*} 

Note that $\{\mc{S}\}$ partitions the jobs into $(k-1)Q/2+1$ partitions such that the size of a big partition is $n(nQ-c) \geq n(n-1)Q$ and the size of a small partition is $n$. Let $f_i$ be the first time that a $(1-\delta)$ fraction of the  jobs in $\mc{S}_i$ is completely executed, and let $s_i$ be the first time that more than $\delta$ fraction of the jobs in $\mc{S}_i$ is started. Because of the expansion property of the NO Case, we can not start more that $\delta$ fraction of the jobs in the second partition, before finishing at least $1-\delta$ fraction of the jobs in the first partition. This implies that $f_1\leq s_2$. Similarly, $f_1+1 \leq s_3$ and $f_1+Q-2 \leq s_Q$ . The same inequalities hold for any big partition and the small partitions following it. This means that, beside $\delta$ fraction of the jobs in the $i$-th and $(i+1)$-th big partitions, the rest of the jobs in the $(i+1)$-th big partition start $Q-1$ steps after finishing the jobs in the $i$-th big partition. Also we need at least $\frac{(1-\delta)n(n-1)Q}{(1+Q\epsilon)n^2}=(1-\epsilon_1)Q$ time to finish $1-\delta$ fraction of the jobs in a big partition. This gives that the makespan is  at least:
\begin{align*}
(1-\epsilon_1)(k+1)Q/2+(k-1)(Q-1)/2 \geq (1-\epsilon_2)kQ
\end{align*} 
where $\epsilon_2 = \epsilon_2(Q,k,\epsilon,\delta)$, which can be made small enough for an appropriate choice of $Q,k,\epsilon$ and $\delta$. %One can see that with choosing them wisely we can make sure that $\epsilon_2$ is as small as we want. Combining the completeness and soundness we get a $2-\epsilon_2$ hardness of approximation for \Pprecpmtn.
\end{proof}
\section{The Perfect Matching Approach}
\label{sec:PerfectMatching}
In this section, we prove Theorem~\ref{thm:BKkpartite} by presenting a direct reduction from a bipartite graph of Theorem~\ref{thm:BK2} to a $k$-partite graph. It is also proved in~\cite{BK09} that Theorem~\ref{thm:BK2} holds assuming a variant of the \uniquegames Conjecture, and note that the former implies Theorem~\ref{thm:BK}. 
%is also NP-hard assuming (a variant of) \uniquegames (Hypothesis~\ref{hyp:bk09}).  
 \begin{theorem}
\label{thm:BK2}
For every $\epsilon, \delta >0$, and positive integer $Q$, the following problem is NP-hard assuming Hypothesis~\ref{hyp:bk09}: given an n-by-n bipartite graph $G = (V,W,E)$,  distinguish  between the following two cases: \begin{itemize}
\item YES Case: We can partition $V$ into disjoints sets $V_0,V_1,\dots,V_{Q-1}, V_{err}$ with $|V_0|=|V_1|=\dots=|V_{Q-1}| = \frac{1-\epsilon}{Q}$, and $W$ into disjoint sets $W_0,W_1,\dots,$ $W_{Q-1},W_{err}$ with $|W_0|=|W_1|=\dots =|W_{Q-1}| =\frac{1-\epsilon}{Q}$ such that for every $i \in [Q]$ and any vertex $w \in W_i$, $w$ only have edges to vertices in $V_i \cup V_{err}$.
%\item YES Case: We can partition $V$ into disjoints sets $V_1,V_2,\dots,V_Q$ with $|V_1|,|V_2|,\dots,|V_Q| \geq \frac{1-\epsilon}{Q}$, and $W$ into disjoint sets $W_1,W_2,\dots,W_Q$ with $|W_1|, |W_2|,\dots, |W_Q|  \geq \frac{1-\epsilon}{Q}$ such that for every $i \in \{1,2,\dots,Q\}$ and any vertex $w \in W_i$, $w$ only have edges to vertices in $V_j$ for $j \leq i$.
\item NO Case: For any $S\subseteq V$, $T \subseteq W$, $|S| = \delta n$, $|T| = \delta n$, there is an edge between $S$ and $T$.
\end{itemize}
\end{theorem}
\subsubsection{Reduction}
We present a reduction from an $n$-by-$n$ bipartite graph $G=(V,W,E)$ to a $k$-partite graph $G_k=(U_1, ..., U_k, E_1,..., E_{k-1})$. From the expansion property of the No Case in Theorem~\ref{thm:BK2}, we get that the size of the maximum matching is at least $(1-\delta)n$. Therefore we can assume that the graph $G$ has a matching of size at least $(1-\delta)n$. We find a maximum matching $M$ and remove all the other vertices from $G$.  Let the resulting graph be $G'=(V',W',E')$, where $|V'| = |W'| = n' \geq n(1-\delta)$, $V'= \{v_0,\dots,v_{n'-1}\}$ and $W'=\{w_0,\dots,w_{n'-1}\}$.  Also assume w.l.o.g. that $v_i$ is matched to $w_i$ in the matching $M$ for all $i \in [n']$.
\begin{observation}
\label{obs:Yesk}Assume that $G$ satisfies the YES Case of Theorem~\ref{thm:BK2}, and let $\{V_{i}\},\{W_{i}\}$ for $i \in [Q]$ and $V_{err}$, $W_{err}$  be the \emph{good} partitioning. We use the latter to define a new partitioning $\{V'_i\},\{W'_i\}$ for $i \in [Q]$ , and $V'_{err}, W'_{err}$ as follows: \begin{align*}
V'_{err} & := ( V_{err} \cup \{v_i | w_i \in W_{err}\} ) \cap V' &&&V'_i &:= ( V_i \backslash V'_{err} )\cap V' \\
W'_{err} & := ( W_{err} \cup \{w_i | v_i \in V_{err}\} ) \cap W' &&& W'_i &:= ( W_i \backslash W'_{err} )\cap W'
\end{align*}
then the following two observations hold: \begin{itemize}
\item For all $i \in [Q]$, $|V'_i| \geq (1-\delta-2\epsilon)n$, $|W'_i| \geq (1-\delta-2\epsilon)n$.
\item for any vertex $w \in W'_i$, $w$ only have edges to vertices in $V'_i \cup V'_{err}$
\end{itemize}
%Let $V'_{err} = ( V_{err} \cup \{v_i | w_i \in W_{err}\} ) \cap V'$, $V'_i := ( V_i \backslash V'_{err} )\cap V'$ for all $i \in [Q]$. Similarly, let $W'_{err} = ( W_{err} \cup \{w_i | v_i \in V_{err}\} ) \cap W'$, $W'_i := ( W_i \backslash W'_{err} )\cap W'$  for all $i \in [Q]$, the following two observation holds:
%\begin{itemize}
%\item for all $i \in [Q]$, $|V'_i| \geq (1-\delta-2\epsilon)n$, $|W'_i| \geq (1-\delta-2\epsilon)n$,  and any vertex $w \in W'_i$, $w$ only have edges to vertices in $V'_i \cup V'_{err}$.
%\end{itemize}
\end{observation}
\begin{observation}
\label{obs:Nok}
Assume that $G$ satisfies the NO Case of Theorem~\ref{thm:BK2}, then $G'$ satisfies the NO Case as well, i.e., for any two sets $S\subseteq V'$, $T \subseteq W'$, $|S|=\delta n$, $|T|=\delta n$, there is an edge between $S$ and $T$.
\end{observation}

We are now ready to construct the $k$-partite graph $G_k$ from $G'$.
\begin{itemize}
\item Let $U_i=\{u_{i,0},u_{i,1},\dots,u_{i,n'-1}\}$ be a set of vertices of size $n'$ for all $i \in \{1,\dots,k\}$.
\item For any edge $e=(v_i,v_j) \in E'$ add edge $(u_{l,i},u_{l+1,j})$ to  $E_l$, for all $l \in \{1,\dots,k-1\}$.
\end{itemize}
\subsubsection{Completeness}
We show that if the given bipartite $G$ satisfies the properties of the YES Case of Theorem~\ref{thm:BK2}, then the YES Case of Theorem ~\ref{thm:BKkpartite} holds. 
Hence assume that we are in the YES Case and let $\{V_{i}\}$ for $i \in [Q]$ denote the \emph{good} partitioning and $\{V'_{i}\}$ denote the partitioning derived from it as described in Observation~\ref{obs:Yesk}. For all $ l = \{1,\dots,k\}$ and $i \in [Q]$ let
\begin{align*}
U_{l,i} &= \{u_{l,j}| v_j \in V'_i\}\\
U_{1,err} &= \{u_{1,j}| v_j \in V'_{err}\}
\end{align*}
It follows from Observation~\ref{obs:Yesk} and the fact that we have the same set of edges in all the layers, that the new partitioning has the properties of the YES Case of Theorem~\ref{thm:BKkpartite}.

\subsubsection{Soundness}
We show that if the given bipartite $G$ satisfies the properties of the NO Case of Theorem~\ref{thm:BK2}, then the YES Case of Theorem~\ref{thm:BKkpartite} holds. To that end, assume that we are in the NO Case, therefore the given bipartite graph satisfy that for any $S \subseteq V$, $T\subseteq W$, $|S| = |T| = \delta n$, there is an edge between $S$ and $T$. From Observation~\ref{obs:Nok} we get that the same expansion property holds for $G'$, i.e. for any $S \subseteq V$, $T\subseteq W$, $|S| = |T| = \delta n$, there is an edge between $S$ and $T$. Moreover, we have the same set of the edges in all the layers, so we get that each layer has the expansion property.

\section{Hardness of Approximation}
\label{UGLB}
In this section, we prove Theorem~\ref{thm:main22} of Section~\ref{sec:Pprecpmtn}. For the sake of presentation, we restate Theorem~\ref{thm:BK} as it is a key component in the reduction. In other words, we prove that assuming the variant of the \uniquegames Conjecture in~\cite{BK09}, it is NP-hard to approximate the scheduling problem \Pprecpmtn within any constant factor strictly better than $3/2$. To do so, we present a reduction from a bipartite graph of Theorem~\ref{thm:BK88} to a scheduling instance $\tilde{\mc{I}}$ such that: \begin{itemize}
\item If $G$ satisfies the YES Case of Theorem~\ref{thm:BK88}, then $\tilde{\mc{I}}$ has a schedule whose makespan is roughly $2Q$.
\item If $G$ satisfies the NO Case of Theorem~\ref{thm:BK88}, then every schedule for $\tilde{\mc{I}}$ must have a makespan of at least $3Q - 1 -\epsilon Q$.
\end{itemize}

%In this section we present a $3/2-\epsilon$ hardness of approximation for \Pprecpmtn assuming the following (implicit) result of \cite{BK09}.
\begin{theorem}
\label{thm:BK88}
For every $\epsilon, \delta >0$, and positive integer $Q$, given an n by n bipartite graph $G = (V,W,E)$ such that, assuming a variant of \uniquegames Conjecture, it is NP-hard to distinguish between the following two cases: \begin{itemize}
\item YES Case: We can partition $V$ into disjoints sets $V_0,V_1,\dots,V_{Q-1}, V_{err}$ with $|V_0|=|V_1|=\dots=|V_{Q-1}| = \frac{1-\epsilon}{Q}$, and $W$ into disjoint sets $W_0,W_1,\dots,$ $W_{Q-1},W_{err}$ with $|W_0|=|W_1|=\dots =|W_{Q-1}| =\frac{1-\epsilon}{Q}$ such that for every $i \in [Q]$ and any vertex $w \in W_i$, $w$ only have edges to vertices in $V_i \cup V_{err}$.
%\item YES Case: We can partition $V$ into disjoints sets $V_1,V_2,\dots,V_Q$ with $|V_1|,|V_2|,\dots,|V_Q| \geq \frac{1-\epsilon}{Q}$, and $W$ into disjoint sets $W_1,W_2,\dots,W_Q$ with $|W_1|, |W_2|,\dots, |W_Q|  \geq \frac{1-\epsilon}{Q}$ such that for every $i \in \{1,2,\dots,Q\}$ and any vertex $w \in W_i$, $w$ only have edges to vertices in $V_j$ for $j \leq i$.
\item NO Case: For any $S\subseteq V$, $T \subseteq W$, $|S| = \delta |V|$, $|T| = \delta |W|$, there is an edge between $S$ and $T$.
\end{itemize}
\end{theorem}
%In this section, we reduce a bipartite graph instance of Theorem~\ref{thm:BK} to a scheduling problem instance of \Pprecpmtn, and prove the following result: \begin{theorem}
%\label{thm:main} For any $\epsilon>0$, it is NP-hard to approximate \Pprecpmtn within a factor of $3/2 -\epsilon$, assuming (a variant of) the \uniquegames Conjecture.
%\end{theorem}
\subsubsection{Reduction}
We present a reduction from an $n$-by-$n$ bipartite graph $G=(V,W,$ $E)$ to a scheduling instance $\tilde{\mc{I}}$, for some integer $Q$ that is the constant of Theorem~\ref{thm:BK2}: \begin{itemize}
\item For each vertex $v \in V$, we create a set $\mc{J}_v$ of $Qn$ jobs each of size 1, and let $\mc{J}_V: = \bigcup_{v\in V} \mc{J}_v$.
\item For each vertex $w \in W$, we create a set $\mc{J}_w$ of $Q(n+1)-1$ jobs \begin{align*}
\mc{J}_w = \{J^1_w,J^2_w,\dots, J^{Q-1}_w\} \cup \mc{J}_w^Q
\end{align*} 
where $\mc{J}_w^Q$ is the set of the last $Qn$ jobs, and the first $Q-1$ jobs are the \emph{chain jobs}. We also define $\mc{J}_W$ to be $\mc{J}_W := \bigcup_{w\in W} \mc{J}_w^Q$.
\item For each edge $e = (v,w) \in E$, we have a precedence constraint between $J_v \prec J_w^1 $ for all $J_v \in \mc{J}_v$.
\item For each $w\in W$, we have the following precedence constraints: \begin{align*}
J_w^i \prec J^{i+1}_w && \forall i \in \{1,2,\dots, Q-2\} \\
J^{Q-1}_w \prec \mc{J}_w ^Q
\end{align*}
\end{itemize} 

In total, the number of jobs and precedence constraints is polynomial in $n$ since \begin{align*}
\text{number of the jobs} & \leq Qn^2 + n(Q(n+1)-1) = 2Qn^2 + Qn -n 
\end{align*}
For a subset $\mc{S}$ of jobs  in our scheduling instance $\tilde{\mc{I}}$, we denote by $\Psi(\mc{S}) \subseteq V \cup W$ the set of their \emph{representative} vertices in the starting graph $G$. Similarly, for a subset $S \subseteq V\cup W$, $\Psi^{-1}(S) \subseteq \mc{J}_V \cup \mc{J}_W$ is the set of \emph{all} jobs, except for chain jobs, corresponding to vertices in $S$,i.e., \begin{align*}
\Psi^{-1}(S) =\left( \bigcup_{v \in \left(S \backslash W \right)} \mc{J}_v \right) \cup \left( \bigcup_{w \in \left(S \backslash V \right)} \mc{J}_w^Q \right)
\end{align*}
A subset $\mc{S}$ of jobs with $S = \Psi(\mc{S})$ is said to be \emph{complete} if $\mc{S} = \Psi^{-1}(S)$.

W.l.o.g. assume that $Q$ divides $n$. Finally the number of machines is $n^2$. Before proceeding with the proof of Theorem~\ref{thm:main22}, we record the following easy observations: \begin{observation}
\label{obs:obs1}
If for some $w\in W$, there exist a feasible schedule $\sigma$ in which a job $J \in \mc{J}_w^Q$ starts before time $T$, then the set $\mc{A} \subseteq \mc{J}_V$ of all its predecessors in $\mc{J}_V$ must have finished executing in $\sigma$ prior to time $T-Q$. Moreover $\mc{A}$ is complete, i.e., $\mc{A}=\mc{J}_V$
\end{observation}
\begin{observation}
\label{obs:obs2}
For any subset $\mc{A} \subseteq \mc{J}_V \cup \mc{J}_W$, we have that \begin{align*}
\left|\Psi \left( \mc{A} \right) \right| \geq \frac{\left| \mc{A}\right|}{nQ}
\end{align*} 
where the bound is met with equality if $\mc{A}$ is complete.
\end{observation}
\subsubsection{Completeness}
\label{sec:comp}
Let $V_0,V_1,\dots,V_{Q-1},V_{err},W_0,W_1,\dots,W_{Q-1},W_{err}$ be the partitions as in the YES Case of  Theorem~\ref{thm:BK88}. For ease of notation, we merge $V_{err}$ with $V_0$, and $W_{err}$ with $W_{Q-1}$, i.e., $V_0 \leftarrow V_0\cup V_{err}$ and $W_{Q-1}\leftarrow W_{Q-1} \cup W_{err}$.

%\begin{align*}
%|V_1| = |V_2| = \dots = |V_Q| \geq \frac{1-\epsilon}{Q} \\
%|W_1|  = |W_2| = \dots = |W_Q| \geq \frac{1-2\epsilon}{Q}
%\end{align*}
Note that this implies that for all $i\in [Q]$, any vertex $w \in W_i$ is only connected to vertices in $V_j$ where $j \leq i$ also: \begin{align}
\label{partsizes}
|W_i|,|V_i| \leq \left(\frac{1-\epsilon}{Q}  + \epsilon \right)n &&& \forall i\in[Q]
\end{align}

For a subset $S \subseteq V \cup W$, we denote by $\mc{J}_S$ the set of jobs corresponding to vertices in $S$, i.e., $\mc{J}_S = \cup_{u \in S} \mc{J}_u$. Also, for an index $i \in \{1,2,\dots,2Q\}$, we define a job set $\mc{T}_{i}$ as follows: \begin{align*}
\mc{T}_i = \left\{ 
  \begin{array}{l l}
    \mc{S}_i \cup \mc{J}_{V_{i-1}}  & 0 \leq i < Q\\
    \mc{S}_i \cup  \mc{J}_{W_{i-Q - 1}}^Q & Q \leq i < 2Q
  \end{array} \right.
\end{align*} 
where \begin{align*}
\mc{S}_i = \{J_{W_{k-1}}^\ell: 1 \leq \ell < Q, \, \, k \in [Q], \text{ and } k+\ell = i \}
\end{align*}
The intuition behind partitioning the jobs into $\mc{S}$ and $\mc{T}$ follows from the same reasoning of the completeness proof of Appendix~\ref{sec:appkprte}. Observe here that using the structure of the graph, we get that if there exists two jobs $J,J'$ such that $J' \prec J$ and $J \in \mc{J}_{W_k}^\ell$, then $J'$ can only be in one of the following two sets: \begin{align*}
J' \in \mc{J}_{V_{k'}} \text{ s.t. }k' \leq k && \text{ or } && J'\in \mc{J}_{W_{k'}}^{\ell'} \text{ s.t. } k' = k, \ell' < \ell
\end{align*}
This then implies  that a schedule $\sigma$ in which we first schedule $\mc{T}_1$ then $\mc{T}_2$, and so on up to $\mc{T}_{2Q}$ is indeed a valid schedule. Now using equation (\ref{partsizes}) and the construction of our scheduling instance $\tilde{\mc{I}}$, we get that \begin{align*}
|\mc{T}_i| & \leq   Qn^2 \left(\frac{1-\epsilon}{Q}  + \epsilon \right) + nQ \left(\frac{1-\epsilon}{Q}  + \epsilon \right) \\
& \leq n(n+1)(1 + \epsilon Q)
\end{align*}
Hence the total makespan of $\sigma$ is at most \begin{align*}
\sum_{i=1}^{2Q} \frac{|\mc{T}_i|}{n^2}  = 2Q \left(1 + \epsilon Q + O\left(\frac{1}{n}\right)\right)
\end{align*}
which tends to $2Q + \epsilon'$ for large values of $n$.

\subsubsection{Soundness}
Assume towards contradiction that there exists a schedule for $\tilde{\mc{I}}$ with a maximum makespan less than $t:=3Q - 1 - 2 \epsilon Q$, and let $\mc{A}$ be the set of jobs in $\mc{J}_W$ that started executing by or before time $s:=2Q - 1-  \epsilon Q$, and denote by $\mc{B}$ the set of their predecessors in $\mc{J}_V$. Note that $\mc{B}$ is \emph{complete} by Observation~\ref{obs:obs1}. Now since $t-s =Q-\epsilon Q$, we get that $|\mc{A}| \geq \epsilon Q n^2$, and hence, by Observation~\ref{obs:obs2}, $|A|: = |\Psi(\mc{A})| \geq \frac{\epsilon Q n^2}{Q n} = \epsilon n$. 
Applying Observation~\ref{obs:obs1} one more time, we get that all the jobs in $\mc{B}$ must have finished executing in $\sigma$ by time $Q -\epsilon Q$, and hence $|\mc{B}| \leq Qn^2(1-\epsilon)$. Using the fact that $\mc{B}$ is complete, we get that $|B|:= |\Psi(\mc{B})| \leq \frac{Qn^2(1-\epsilon)}{Qn}= (1-\epsilon)n$, which contradicts with the NO Case of Theorem~\ref{thm:BK}.

It is important to note here that we can settle for a weaker structure of the graph corresponding to the completeness case of Theorem~\ref{thm:BK}. In fact, we can use a graph resulting from Theorem 2 in \cite{Ola11}, and yet get a hardness of $3/2 - \epsilon$. This will then yield this somehow stronger statement: \begin{theorem}
For any $\epsilon >0$, and $\eta \geq \eta(\epsilon)$, where $\eta(\epsilon)$ tends to 0 as $\epsilon$ tends to 0, if \Oprec has no $(2-\epsilon)$-approximation algorithm, then \Pprecpmtn has no $(3/2 - \eta)$-approximation algorithm.
\end{theorem}

\section{Linear Programming Formulation for \Pprecpmtn}
\label{sec:LP}
In this section, we will be interested in a feasibility Linear Program, that we denote by {\bf [LP]}, for the scheduling problem \Pprecpmtn. For a makespan \emph{guess} $T$, {\bf [LP]} has a set of indicator variables $\{x_{j,t}\}$ for $j\in \{1,2,\dots,n\}$ and $t \in \{1,2,\dots,T\}$. A variable $x_{j,t}$ is intended to be the fraction of the job $J_j$ scheduled between time $t-1$ and $t$. The optimal makespan $T^*$ is then obtained by doing a binary search and checking at each step if {\bf [LP]} is feasible: \begin{align}
\label{csmachines}
 & \sum_{j=1}^n x_{j,t} \leq m & \forall t \in \{1,2,\dots,T\}\\
\label{cscompletion}
& \sum_{t=1}^T x_{j,t} = 1 & \forall j \in\{1,2,\dots,n\} \\
\label{cspredence}
& \sum_{t=1}^{t'-1}x_{\ell,t} + \sum_{t = t'+1}^{T}x_{k,t} \geq 1 & \forall J_\ell \prec J_k, \forall t'\in \{1,2,\dots,T\} \\
%\label{nonneg}
& x_{j,t} \geq 0 & \forall j \in\{1,2,\dots,n\}, \forall t\in \{1,2,\dots,T\} \nonumber
\end{align}
To see that {\bf [LP]}  is a valid relaxation for the scheduling problem \Pprecpmtn, note that constraint (\ref{csmachines}) guarantees that the number of jobs processed at each time unit is at most the number of machines, and constraint (\ref{cscompletion}) says that in any feasible schedule, all the jobs must be assigned. Also any schedule that satisfies the precedence requirements must satisfy constraint (\ref{cspredence}).
%We show in Section~\ref{sec:LPgap} that {\bf [LP]} has an integrality of 2.
\subsection{Integrality Gap}
\label{sec:LPgap}
In order to show that  {\bf [LP]} has an integrality gap of 2, we start by constructing a family of integrality gap instances of $3/2$ and gradually increase this gap to 2. The reason is that the $3/2$ case captures the intrinsic hardness of the problem, and we show how to use it as basic building block for  the construction of the target integrality gap instance of 2.
\paragraph{Basic Building Block}

We start by constructing an \Pprecpmtn scheduling instance $\mc{I}(d)$ parametrised by a large constant $d\geq 2$, that shows that  the integrality gap of {\bf [LP]} is $3/2$, and constitutes our main building block for the next reduction.  Let $m$ be the number of machines, and $n = 2dm - (d-1)$ the number of jobs. The instance $\mc{I}(d)$ is  then constructed as follows:
\begin{itemize}
\item The first $dm-(d-1)$ jobs $J_1,\dots,J_{dm-(d-1)}$ have no predecessors [Layer 1].
\item A \emph{chain} of $(d-1)$ jobs $J_{dm-(d-1)+1},\dots,J_{dm}$ such that $J_{dm-(d-1)+1}$ is the successor of all the jobs in the Layer 1, and $J_{k-1} \prec J_{k}$ for $k \in \{dm-(d-1)+2,\dots,dm\}$ [Layer 2].
\item The last $dm - (d-1)$ jobs $J_{dm+1},\dots,J_{2dm - (d-1)}$ are successors of $J_{dm}$ [Layer 3]. 
\end{itemize} 
We first show that $\mc{I}(d)$ is an integrality gap of $3/2$ for {\bf [LP]}. This basically follows from the following lemma:
\begin{lemma}
\label{lemma:firstgap}
Any feasible schedule for $\mc{I}(d)$ has a makespan of at least $3d-2$, however {\bf [LP]} has a feasible solution $\{x_{j,t}\}$, for $t \in \{1,2,\dots,2d\}$ and $ j= 1,\dots,2dm - (d-1)$ of value $2d$. Moreover, for $t = d + 1 + \ell$, $\ell \in \{1,2,\dots,d-1\}$, the machines in the feasible LP solution can still execute a load of $\frac{\ell }{d}$, i.e., $m - \sum_{J_j \in \mc{J}} x_{j,t} \geq \frac{\ell}{d}$.
\end{lemma}
\begin{proof}
Consider the following fractional solution: \begin{align*}
\text{[Layer 1]} \quad \quad &x_{1,1}=1 \quad \quad \quad \text{ and } \quad \quad \quad x_{j,t} = \frac{1}{d} \quad \quad\\ &\forall j\in \{2,3,\dots, dm-(d-1)\}, \forall t \in \{1,2,\dots,d\}\\
\text{[Layer 2]} \quad \quad &x_{dm-(d-1)+1,\ell+1} = x_{dm-(d-1)+2,\ell+2} = \dots =  x_{dm,\ell+d-1}  = \frac{1}{d} \quad \quad \quad \,\,\,\,\\& \forall \ell \in \{1,2,\dots, d\}\\
\text{[Layer 3]} \quad \quad &x_{j,t} = \frac{1}{d} \quad \quad \quad \forall j \in \{dm + 1,dm+2,\dots,2dm - (d-1)\}, \\& \forall t\in \{d+1,d+2,\dots,2d\}
\end{align*}
One can easily verify that each job $J$ is completely scheduled, i.e., $\sum_{t=1}^{2d}x_{J,t} = 1$. Moreover, the workload at each time step is at most $m$. To see this, we consider the following three types of time steps:\begin{enumerate}
\item For $t=1$, the workload is \begin{align*}
1 +\frac{1}{d} \times \left(dm-(d-1)-1\right) = m
\end{align*}
\item For $t=2,\dots,d$, the workload is \begin{align*}
\frac{1}{d} \times \left((dm - (d-1)-1) + t-1 \right) \leq m
\end{align*}
\item For $t = d+1,\dots,2d$, the workload is \begin{align*}
\frac{1}{d}\left((dm - (d-1)) + (2d - t) \right) = m - \frac{t - (d+1)}{d}
\end{align*}
\end{enumerate}
Note that in this feasible solution, we have that for $t = d + 1 + i$, $i\in \{1,2,\dots,d-1\}$, the machines can still execute a load of $\frac{i}{d}$. 

We have thus far verified that $\{x_{j,t}\}$ satisfies  the constraints (\ref{csmachines}) and (\ref{cscompletion}) of {\bf [LP]}. Hence it remains to check (\ref{cspredence}). Except for job $J_1$, any two jobs $J_k$ and $J_\ell$, such that $J_k$ is a $\emph{direct}$ predecessor of $J_\ell$, satisfy the following properties by construction: If $t_k = \min \left\{ t: x_{k,t}>0\right\}$, then \begin{enumerate}
\item $\max \left\{t: x_{k,t}>0 \right\} = t_k+d-1$.
\item $\min \left\{t: x_{\ell,t}>0 \right\} = t_k + 1$.
\item $\max \left\{t: x_{\ell,t}>0 \right\} = t_k + d $.
\end{enumerate}
Hence for any such jobs $J_k$ and $J_\ell$, and for any $t \in \{t_k,t_k+1,\dots,t_k+d\}$ we get \begin{align*}
\sum_{\tilde{t}=1}^{t-1} x_{k,\tilde{t}} + \sum_{\tilde{t} = t+1}^{2d} x_{\ell,\tilde{t}} & = \sum_{\tilde{t}=t_k}^{t-1} x_{k,\tilde{t}} + \sum_{\tilde{t} = t+1}^{t_k +d} x_{\ell,\tilde{t}} \\
& =\frac{t - 1 - t_k + 1}{d } + \frac{t_k + d  - (t+1) + 1}{d} \\
& = 1
\end{align*}
Similarly, for $t \in \{ 1,2,\dots,t_k-1\}$ (respectively $t\in \{t_k+d+1,\dots,2d\}$), the second (and respectively first) summation will be 1. 

On the other hand, one can see that we should schedule all the jobs in Layer 1 in order to start with the first job in Layer 2. Similarly, due to the \emph{chain-like} structure of Layer 2, it requires $d-1$ times steps to be scheduled, before any job in Layer 3 can start executing. Hence the makespan of any feasible schedule is at least \begin{align*}
\frac{dm -(d-1) }{m} + (d-1) + \frac{dm - (d-1)}{m} =  3d - \frac{2(d-1)}{m} - 1 \geq 3d -2.
\end{align*}

\end{proof}
\paragraph{Final Instance}
We now construct our final integrality gap instance $\mc{I}(k,d)$, using the basic building block $\mc{I}(d)$. This is basically done by \emph{replicating} the structure of $\mc{I}(d)$, and arguing that any feasible schedule for $\mc{I}(k,d)$ must have a makespan of roughly $2kd$, whereas we can \emph{extend} the the LP solution of Lemma~\ref{lemma:firstgap} for the instance $\mc{I}(d)$, to a feasible LP solution for $\mc{I}(k,d)$ of value $(k+1)d$. A key point that we use here is that the structure of the LP solution of Lemma~\ref{lemma:firstgap} enables us to schedule  a fraction of the \emph{chain} jobs of a layer, while executing the \emph{non-chain} jobs of the previous layer. We now proceed to prove that the integrality gap of {\bf [LP]} is 2, by constructing a family $\mc{I}(k,d)$ of scheduling instances, using the basic building block $\mc{I}(d)$.
\begin{theorem}
\label{thm:LPgap}
{\bf [LP]} has an integrality gap of 2.%$c$, where $c$ is a constant arbitrarily close to $2$.
\end{theorem}
\begin{proof}
Consider the following family of instances $\mathcal{I}(k,d)$ for constant integers $k$ and $d$, constructed as follows: \begin{itemize}
\item We have $k+1$ layers $\{\mc{L}^1_1, \mc{L}^1_2,\dots, \mc{L}^1_{k+1}\}$ similar to Layer 1 in $\mathcal{I}(d)$, and $k$ layers $\{\mc{L}^2_1, \mc{L}^2_2,\dots, \mc{L}^2_{k}\}$ similar to Layer 2. 
i.e., $\mc{L}^1_i$  has $dm-(d-1)$ jobs $J_{i,1}^1,\dots,J_{i,dm-(d-1)}^1$ for all $i \in \{1,\dots,k+1\}$ and $\mc{L}^2_i$  has $(d-1)$ jobs $J_{i,1}^2,\dots,J_{i,(d-1)}^2$ for all $i \in \{1,\dots,k+1\}$.
\item For $i \in \{1,2,\dots,k\}$: \begin{itemize}
\item Connect $\mc{L}_{i}^1$ to $\mc{L}_{i}^2$ in the same way that Layer 1 is connected to Layer 2 in $\mathcal{I}(d)$, that is, the job $J_{i,1}^2 \in \mc{L}_i^2$ is a successor for all the jobs in $\mc{L}_{i}^1$.
\item Connect $\mc{L}_{i}^2$ to $\mc{L}_{i+1}^1$ in the same way that Layer 2 is connected to Layer 3 in $\mathcal{I}(d)$, that is, all the jobs in $\mc{L}_{i+1}^1$ are successors for the job $J_{i,(d-1)}^2 \in \mc{L}_{i}^2$.
\end{itemize}
\end{itemize}
Notice that for $k=1$, the scheduling instance $\mc{I}(1,d)$ is the same as the previously defined instance $\mc{I}(d)$.
%Notice that $\mc{I}(d)$ previously defined corresponds to  $\mc{I}(1,d)$. 
In any feasible schedule, we need to first schedule the jobs in $\mc{L}_1^1$, then those in $\mc{L}_1^2$, then $\mc{L}_2^1$, and so on, until $\mc{L}_{k+1}^1$. Hence the makespan of any such schedule is at least \begin{align*}
(k+1) \frac{dm-(d-1)}{m} + k (d-1) > 2kd + d - k -1.
\end{align*}
We now show that {\bf [LP]} has a feasible solution of value $(k+1)d$. Let $\{x_{j,t}\}$ for $t = 1,\dots,2d$ and $ j = 1,\dots,2dm-(d-1)$ be the feasible solution of value $2d$ obtained in Lemma~\ref{lemma:firstgap}. It would be easier to think of $\{x_{j,t}\}$ as $\{x^1_{j,t}\}\cup\{x^2_{j,t}\}\cup\{x^3_{j,t}\}$  where for $i=1,2,3$,  $\{x^1_{j,t}\}$ is the set of LP variables corresponding to variables in Layer $i$ in $\mc{I}(d)$. We now construct a feasible solution $\{y_{j,t}\}$ for $\mc{I}(k,d)$. We similarly think of $\{y_{j,t}\}$ as $\{y_{j,t}^1\} \cup \{y_{j,t}^2\}$, where $y_{j,t}^i$ is the set of LP variables corresponding to jobs in $\mc{L}^i_\ell$, for some $1\leq \ell \leq k+1$. The set $\{y^i_{j,t}\}$ for $\mc{I}(k,d)$ can then be readily constructed as follows: \begin{itemize}
\item for $J_{1,j}^1 \in \mc{L}_1^1$, $y^1_{j,t} = x^{1}_{j,t}$ for $t \leq 2d$, and 0 otherwise.
\item for $J_{i,j}^2 \in \mc{L}_i^2$, $y^2_{j,t +(i-1)d } = x^{2}_{j,t}$ for $i = 1,2,\dots,k$, and $t \leq 2d$, and 0 otherwise.
\item for $J_{i,j}^1 \in \mc{L}_i^1$, $y^1_{j,t +(i-1)d } = x^{3}_{j,t}$ for $i = 2,3,\dots,k+1$, and $t \leq 2d$, and 0 otherwise.
%\item for $j \in \mc{L}_i^1$, $y_{\phi_i(j),t +(i-1)d } = x^{2,3}_{k,t}$ for $i = 2,3,\dots,k+1$, and $t \leq 2d$, and 0 otherwise.
\end{itemize}
Using Lemma~\ref{lemma:firstgap}, we get that for $t = d + 1 + i$, $i\in \{1,2,\dots,d-1\}$, the machines in the feasible LP solution can still execute a load of $\frac{i}{d}$, and hence invoking the same analysis of Lemma~\ref{lemma:firstgap} with the aforementioned observation for every two consecutive layers  of jobs, we get that $\{y_{j,t}\}$ is a feasible solution for {\bf [LP]} of value $(k+1)d$.
\end{proof}

\begin{figure}[h!]
  \centering
    \includegraphics[width=1\textwidth]{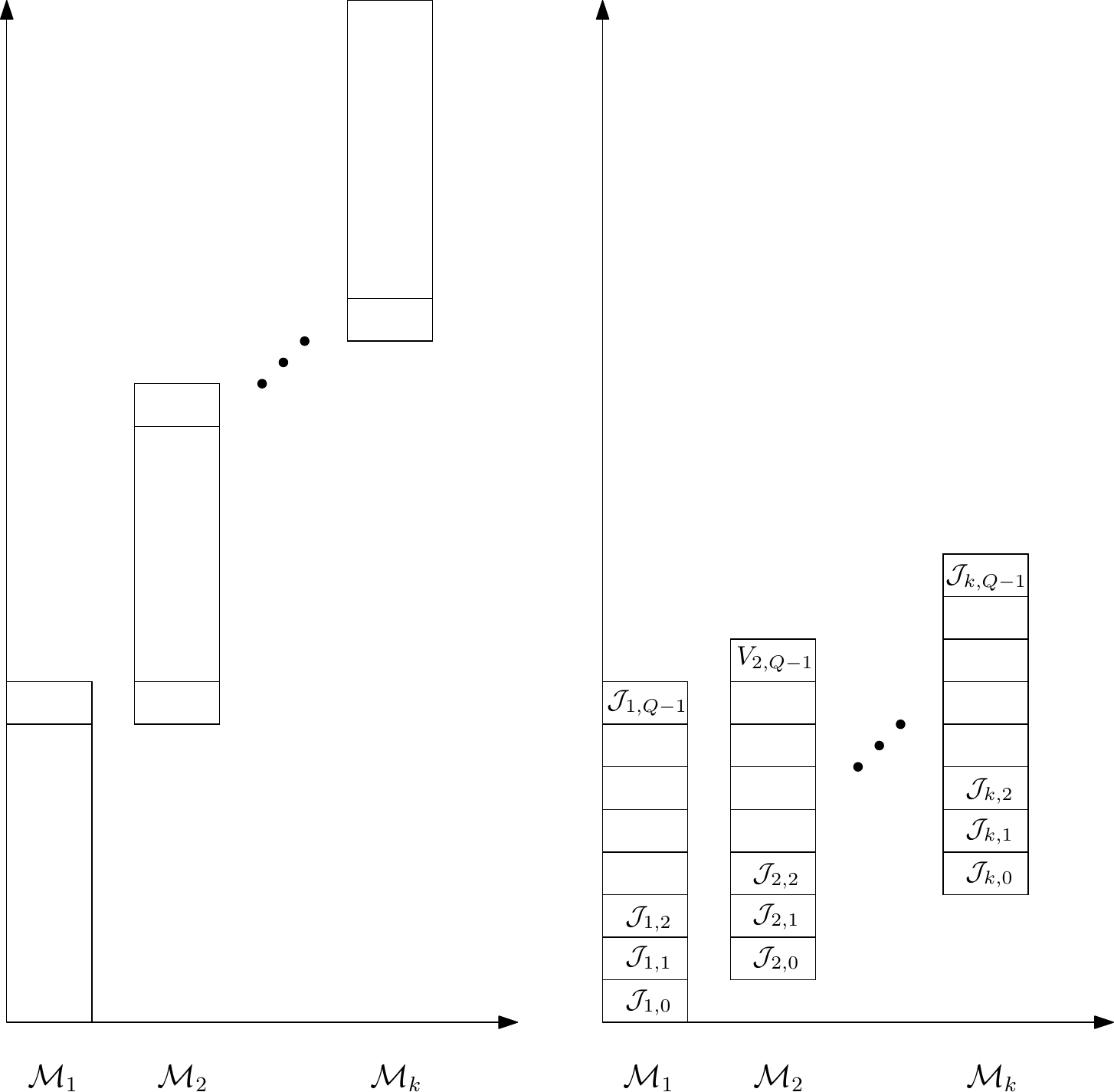}
      \caption{Structure of the Soundness Versus Completeness of \Qprec Assuming Hypothesis~\ref{hyp:kpartite}.\emph{The schedule on the left corresponds to the case where the graph represents the NO Case of Hypothesis~\ref{hyp:kpartite}; note that most of the machines are idle but for a small fraction of times. The schedule on the right corresponds to the case where the graph represents the YES Case; the schedule is almost packed. This case also illustrates the ordering of the jobs within each machine according to the partitioning of the jobs in the $k$-partite graph.}}
      \label{fig:difspeed}
\end{figure}

\begin{figure}[h!]
  \centering
    \includegraphics[width=1\textwidth]{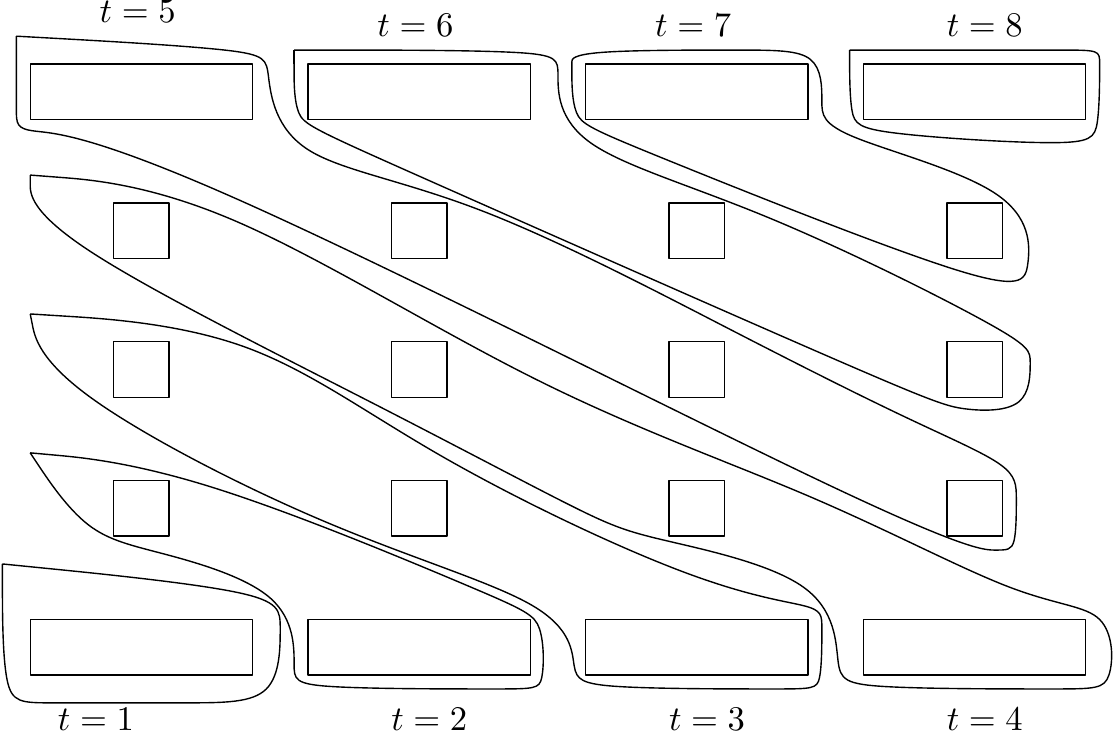}
     \caption{Example of the construction of sets $\{\mc{S}\}$ for \Pprecpmtn in the YES Case of Hypothesis~\ref{hyp:kpartite}, for $k=3$ and $Q=4$, along with their respective finishing time in the defined schedule. \emph{The boxes in the figure represent sets of jobs, and the sets that are grouped together have no precedence constraints within each others. Hence a feasible schedule is to schedule each group during the same time step. These groups corresponds to the sets $\mc{T}_i$ of Appendix~\ref{sec:appkprte}.}}
     \label{fig:UGhard}
\end{figure}

\begin{figure}[h!]
  \centering
    \includegraphics[width=1\textwidth]{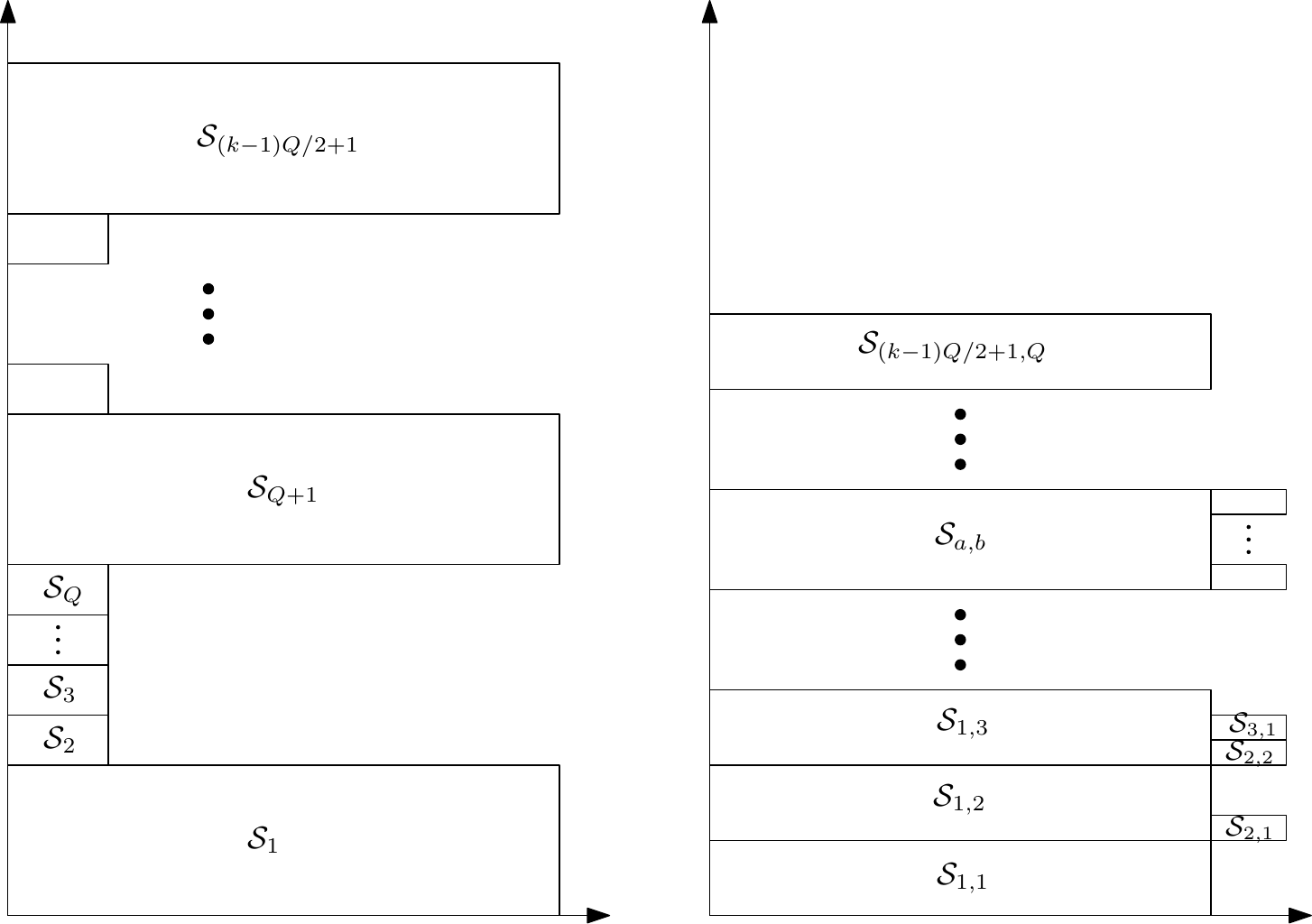}
      \caption{Structure of Soundness versus Completeness of \Pprecpmtn Assuming Hypothesis~\ref{hyp:kpartite}: \emph{The schedule on the left corresponds to the case where the starting graph represents the NO Case of Hypothesis~\ref{hyp:kpartite}; note the most of the machines are idle most of the time in this case. The schedule on the right corresponds to the case where the starting graph represents the YES Case of the hypothesis; Note that the all the machines are packed almost all the time. This case also illustrates our partitioning of the jobs in sets $\{\mc{T}_t\}$, where $\mc{T}_t = \bigcup_{i,j: i+j  - 1 = t } \mc{S}_{i,j}$ }}
  \label{fig:kpartitepmtn}
\end{figure}

\end{document}